\documentclass[onecolumn]{IEEEtran} 
\usepackage{amsthm,amsmath,amssymb,amsfonts,bm}
\usepackage{tikz}
\usepackage{booktabs,multirow,nicematrix}
\usepackage[table,xcdraw]{xcolor}
\usepackage{adjustbox}
\usepackage[shortlabels]{enumitem}
\usepackage[square, comma, sort&compress, numbers]{natbib}
\usepackage{float,graphicx}
\usepackage[colorlinks,linkcolor=blue]{hyperref} 
\newcommand{\rank}{{\mathrm{rank}}}
\newcommand{\wt}{{\mathrm{wt}}}
\newcommand{\hull}{{\mathrm{Hull}}}
\newcommand{\ehull}{{\mathrm{Hull_{E}}}}
\newcommand{\hhull}{{\mathrm{Hull_{H}}}}
\newcommand{\diag}{{\mathrm{diag}}}

\renewcommand{\eqref}[1]{\textcolor{blue}{(\ref{#1})}}

\newtheorem{theorem}{Theorem}[section]
\newtheorem{lemma}[theorem]{Lemma}
\newtheorem{proposition}[theorem]{Proposition}

\newtheorem{definition}[theorem]{Definition}
\newtheorem{example}[theorem]{Example}
\newtheorem{remark}[theorem]{Remark}

\makeatletter
\@addtoreset{equation}{section}
\makeatother

\begin{document}
	\title{The MDS or NMDS for Modified GRS codes with flexible hull dimensions and lengths}
	\author{Zhonghao Liang, Qunying Liao, Jun Zhang, Xiaoping Li
		\thanks{Corresponding author: Qunying Liao. Emails: liangzhongh0807@163.com; qunyingliao@sicnu.edu.cn; junz@cnu.edu.cn; lixiaoping.math@uestc.edu.cn.}
		\thanks{Zhonghao Liang and Qunying Liao are with College of Mathematical Sciences, Sichuan Normal University, Chengdu 610066, China.}
		\thanks{Jun Zhang is with School of Mathematical Sciences, Capital Normal University, Beijing 100048, China.}
		\thanks{Xiaoping Li is with School of Mathematical Science, University of Electronic Science and Technology of China, Chengdu, China.}
		\thanks{This paper is supported by National Natural Science Foundation of China (12471494) and Natural
			Science Foundation of Sichuan Province (2024NSFSC2051).}
		
	}
	\maketitle
	
	\begin{abstract}
Non-generalized Reed–Solomon (in short, non-GRS) type maximum distance separable (in short, MDS), near MDS (in short, NMDS), and linear complementary dual (in short, LCD) codes, as well as the hull of linear codes have interesting practical applications in cryptography and coding theory. In this paper, we focus on a class of non-GRS codes and its extended codes, i.e., modified generalized Reed–Solomon (MGRS) codes and extended MGRS (EMGRS) codes introduced by Wang et al. in 2026. Firstly, we prove that two classes of MGRS codes and EMGRS codes are either MDS or NMDS, derive the necessary and sufficient conditions for these codes to be NMDS, and then completely determine the weight distributions for one class of these NMDS MGRS or NMDS EMGRS codes. Secondly, we construct four classes of MGRS codes which are either Euclidean LCD codes or one-dimensional Euclidean hull codes. Thirdly, we constructively prove that there exist MGRS codes with flexible Hermitian hull dimensions and lengths. In additional, we illustrate the linearly inequivalence of NMDS MGRS codes and elliptic curve NMDS codes by Schur product. Finally, some corresponding examples are given.
	\end{abstract}
	
	\section{Introduction}
	
	Let $\mathbb{F}_{q}$ be a finite field with $q$ elements, $\mathcal{C}$ and $\mathcal{C}^{\perp}$ be an $[n, k, d]$ linear code and its dual code over $\mathbb{F}_{q}$, respectively. Let $A_i$ denote the number of codewords of weight $i$ in $\mathcal{C}$. The polynomial $\sum\limits_{i=0}^{n}A_{i}z^{i}$ is called the weight enumerator of $\mathcal{C}$, and the sequence $(A_0, A_1, ...,A_n)$
	is referred to as the weight distribution of $\mathcal{C}$. Since the weight distribution contains some significant information including its error-correcting capability and error-detection probability, the weight distribution
	of linear codes has always been one of hot topics in theory and practice in recent years \cite{KaiextendNMDS,DingCSNMDS2020,ZhuSX2024NMDS,FanCLNMDS2024,DingYnonRS,HanDCNMDS2023,LiangZHNMDS12026,LiangZHNMDS22026,ZhouHYNMDS2025}. Since NMDS codes
	have important applications in finite projective geometries \cite{AgugliaNMDS,BartoliNMDS},
	designs \cite{DingCSNMDS2020,DingNMDSdesign} and secret sharing schemes \cite{ZhouNMDSscheme}.  And unlike MDS codes \cite{WHFECC}, NMDS codes with the same parameters may have different weight distributions. Therefore, constructing NMDS codes and deciding the weight distributions of NMDS codes are very interesting \cite{ZhuSX2024NMDS,DingCSNMDS2020,FanCLNMDS2024,ZhouYNMDS2025,HanDCNMDS2023,ZhouHYNMDS2025,LiangZHNMDS12026,LiangZHNMDS22026}. 
	
	For a linear code $\mathcal{C}$, the hull is defined by $\hull(\mathcal{C}) = \mathcal{C} \cap \mathcal{C}^{\perp}$.
	It is well-known that the value of $\dim(\hull(\mathcal{C}))$ plays a critical role in determining the computational complexity of algorithms for checking permutation equivalence between linear codes \cite{Sendrierpermutation}.
	Moreover, in quantum coding theory, the number of pre-shared entanglements between the encoder and decoder is closely related to the hull dimension of classical linear codes \cite{EAQECCHull}.
	Of particular interest is the case where the hull is trivial, i.e.,  $\hull(\mathcal{C}) = \{\mathbf{0}\}$, the corresponding code $\mathcal{C}$ is called a linear complementary dual (in short, LCD) code.
	LCD codes have important applications in lattices \cite{HouLCDlattices}, network coding \cite{Braun}, and multi-secret-sharing schemes \cite{AlahmadiLCDmultisecret}.
	These significant applications of the hull have motivated extensive research on computing $\dim(\hull(\mathcal{C}))$ and constructing LCD codes \cite{TangEHLCD,ChenHhull,LuoGJhull,LiCJhull,LuoJQhull,YueQhull,LiangLCDGRL}.
	
	In recent years, the construction of non-GRS type linear codes has attracted considerable attention due to that they can effectively resist the Sidelnikov-Shestakov attack and the Wieschebrink attack \cite{ChennonRSAG,LiangGRL,HuLPTGRS,LuoJQCTGRS,BAKRCTRS,ZhuSXFDGRS,DingYnonRS,ZhuSXTETGRS,ZhangJHRS,LiuMGRS,LiuHWnonRS,LiFnonRScyclic,JinLFnonRS2026,ZhuSXDRL,ZhuSXnonRSMDSNMDS,ZhuSXTCETGRS,LuoJQnonRScurves}. So far, numerous constructions of non-GRS
	MDS codes have been proposed by performing row or column transformations of the generator matrix of GRS codes, such as,
	\begin{itemize}
		\item {\bf Modify the entire row vector:} twisted GRS (in short, TGRS) codes \cite{BeelenTGRS}, A-TGRS codes \cite{ZhaoATGRS}, $(\mathcal{L},\mathcal{P})$-TGRS codes \cite{HuLPTGRS}. 
		\item {\bf Delete the entire row vector:} deleted GRS (in short, DGRS) codes \cite{HanDCDGRS}, generalized DGRS(in short, GDGRS) codes \cite{Ding2025nonGRS}, flexible GDGRS codes \cite{ZhuSXFDGRS}.
		
		\item {\bf Add the column vector:} Roth-Lempel(in short, RL) codes \cite{RL1989}, extended GRS codes \cite{MacWilliamsEGRS}, column TGRS codes \cite{LuoJQCTGRS}, extended RL codes \cite{WuRL3}, generalized RL codes \cite{LiangGRL}.
		
		\item {\bf Modify the entire row vector and delete the entire row vector simultaneously:} Roth-Lempel type DGRS(in short, RL-TGRS) codes \cite{ZhuSXDRL},
		
		\item {\bf Modify the entire row vector and add the column vector simultaneously:} row-column TGRS codes \cite{BAKRCTRS},   the twice-extended TGRS codes \cite{DingYnonRS,ZhuSXTETGRS}, extended RL-TGRS codes \cite{YanTJERLTGRS},   extended TGRS codes \cite{YangSDNONGRS}.

		\item {\bf Add the column vector and delete the entire row vector simultaneously:} extended DGRS codes \cite{ZhuSXnonRSMDSNMDS}, extended GDGRS codes \cite{Ding2025nonGRS}.
	\end{itemize}  
	\iffalse generalized Roth-Lempel codes \cite{LiangGRL}, $(\mathcal{L},\mathcal{P})$-twisted GRS codes \cite{HuLPTGRS}, column TGRS codes \cite{LuoJQCTGRS}, row-column TGRS codes \cite{BAKRCTRS}, deleted GRS codes \cite{ZhuSXFDGRS}, the twice-extended TGRS codes \cite{DingYnonRS,ZhuSXTETGRS}, and so on \cite{ZhangJHRS,LiuMGRS,LiuHWnonRS,ChenHnonRScyclic,JinLFnonRS2026,ZhuSXDRL,ZhuSXnonRSMDSNMDS,ZhuSXTCETGRS,LuoJQnonRScurves}.\fi 
	Most recently, by locally modifying column vectors, Wang, Liu and Luo \cite{LiuMGRS} introduced a class of non-GRS type linear codes and its extended codes called modified GRS (in short, MGRS) codes and extended MGRS (in short, EMGRS ) codes. Their generator matrices differ from those of the GRS codes in only one entry, and both classes of these codes are either MDS or AMDS. They also presented the necessary and sufficient conditions for these codes to be non-GRS MDS codes, and analyzed the interrelationships among several families of non-GRS MDS codes and MGRS codes.
	
	Motivated by the above, in this paper, for MGRS codes and EMGRS codes, we investigate the NMDS property, the weight distributions, the LCD property, and the hulls. Firstly, for two classes of MGRS codes and EMGRS codes, by proving that they are either MDS or NMDS, we present a necessary and sufficient condition for these codes to be NMDS. Meanwhile, we completely determine their weight distributions for a class of NMDS MGRS codes and NMDS EMGRS codes. Secondly, we present four classes of LCD MGRS codes and MGRS codes with one-dimensional Euclidean hull. Thirdly, we constructively prove that there exist MGRS codes with flexible Hermitian hull dimensions. Finally, we show that the MGRS codes are not linearly equivalent to any
	one-point algebraic geometry code $\mathcal{C}(P_1, . . . , P_n, mP_0, E)$ from an elliptic
	curve $E$ by Schur products.
	
	This paper is organized as follows. In Section \ref{sec2}, we recall some  necessary definitions and lemmas. In Section \ref{sec3} (resp. Section \ref{sec4}), we present the definition of the MGRS (resp. EMGRS) code, prove that two classes of MGRS(resp. EMGRS) codes are either MDS or NMDS, and completely determine the weight distribution for a class of NMDS MGRS(resp. EMGRS) codes. In Section \ref{sec5}, we construct four classes of MGRS codes, which are either LCD or one-dimensional Euclidean hull codes. In Section \ref{sec6}, we constructively prove that there exist MGRS codes with flexible Hermitian hull dimensions. In Section \ref{sec7}, we illustrate the linearly
	inequivalence of NMDS MGRS codes and a class of elliptic curve
	NMDS codes. 
	In Section \ref{sec8}, some examples are given. 
	In Section \ref{sec9}, we conclude the whole paper.
	\section{Preliminaries}\label{sec2}
	Throughout this paper, we fix the following notations.
	\begin{itemize} 
		\item $q$ is the odd prime power.
		\item $\mathbb{F}_{q}$(resp. $\mathbb{F}_{q^2}$) is the finite field with $q$(resp. $q^2$) elements and $\mathbb{F}_{q}^{*}=\mathbb{F}_{q}\backslash\left\{0\right\}$. 
		\item $\mathbb{F}_{q}[x]$ is the polynomial ring over $\mathbb{F}_{q}$, and $\mathbb{F}_{q}^{k}[x]=\left\{f(x)\in\mathbb{F}_{q}[x]|\deg(f(x))\leq k-1\right\}$.
		\item For any set $A$, $\# A$ denotes the number of elements in $A$. 
		\item $\mu(\cdot)$ denotes the Möbius function.
		\iffalse \item For any vector $\boldsymbol{\alpha}=\left(\alpha_{1},\ldots,\alpha_{n}\right)$ and the element $x\in\mathbb{F}_{q}$,  $x\boldsymbol{\alpha}$ denotes $\left(x\alpha_{1},\ldots,x\alpha_{n}\right)$.\fi
		\item $\mathbb{Z}_n$ denotes the ring of integers modulo $n$.
		\item $\gcd(a,b)$ denotes the greatest common divisor of two integers $a$ and $b$.
		\item For a matrix $\boldsymbol{A}$, $\boldsymbol{A}^{T_{e}}$ and $\boldsymbol{A}^{T_{h}}$ denotes the transposed
		matrix and the conjugate matrix of $\boldsymbol{A}$, respectively.
	\end{itemize} 
	
	In this section, we review some basic notations and known results about MDS codes, NMDS codes, LCD codes, the hull dimension, the subset problem, and symmetric polynomials, respectively.
	
	\subsection{Some notations and related results for linear codes}
	
	The following Lemma \ref{LCDMDSorNMDSequivalent} provides a necessary and sufficient
	conditions for a linear code to be LCD, or MDS, or NMDS, respectively.
	\begin{lemma}\label{LCDMDSorNMDSequivalent}{\rm( \cite{TangEHLCD}, Proposition 2; \cite{WHFECC}, Theorem 2.4.3; \cite{NMDSsmall}, Definition 2.3)} 	
		Let $\mathcal{C}$ be an $[n, k]$ linear code over $\mathbb{F}_{q}$ with $k\geq 1$. Suppose that  $\boldsymbol{G}$ is the 
		generator matrix of $\mathcal{C}$. Then, the following three statements are true,
		
		$(1)$ $\mathcal{C}$ is LCD if and only if $\rank\left(\boldsymbol{G}\\
		\boldsymbol{G}^{T_{e}}\right)=k$;
		
		$(2)$ $\mathcal{C}$ is $\mathrm{MDS}$ if and only if any $k$ columns of $\boldsymbol{G}$ are $\mathbb{F}_{q}$-linearly independent;	
		
		$(3)$ $\mathcal{C}$ is NMDS if and only if the following three conditions hold simultaneously,
		
		\ \ $(i)$ any $k-1$ columns of $\boldsymbol{G}$ are $\mathbb{F}_{q}$-linearly independent;
		
		\ \ $(ii)$ there exists $k$ columns of $\boldsymbol{G}$ which are $\mathbb{F}_{q}$-linearly dependent; 
		
		\ \ $(iii)$ Any $k+1$ columns of $\boldsymbol{G}$ are full rank.
	\end{lemma}
	
	For an $\left[n,k,d\right]$ linear code, the following Lemma \ref{EHhull} provides a method for calculating the dimension of the Euclidean hull or the Hermitian hull.
	
	\begin{lemma}\label{EHhull}\rm ( \cite{EAQECCHull}, Propositions 3.1-3.2)
		Let $\mathcal{C}_{1}$ and $\mathcal{C}_{2}$ be an $\left[n,k,d\right]_{q}$ linear code with the generator matrix $\boldsymbol{G}_{1}$ and an $\left[n,k,d\right]_{q^2}$ linear code with the generator matrix $\boldsymbol{G}_{2}$ respectively, then, we have
		$$\rank\left(\boldsymbol{G}\boldsymbol{G}^{T_{e}}\right)=k-\dim\left(\mathrm{Hull}_{E}\left(\mathcal{C}^{\perp_{E}}\right)\right)$$and 
		$$\rank\left(\boldsymbol{G}\boldsymbol{G}^{T_{h}}\right)=k-\dim\left(\mathrm{Hull}_{H}\left(\mathcal{C}^{\perp_{H}}\right)\right).$$
	\end{lemma} 
	
	The following Lemmas \ref{Fq*sum}-\ref{kq-1mutaking} are crucial for constructing Euclidean LCD or one-dimensional  Euclidean hull MGRS codes in Section \ref{sec5}. For the convenience,  we fix $\mathbb{F}_{q}^{*}=\langle\gamma\rangle$, $k\mid q-1$ and $\alpha_{i}=\gamma^{\frac{q-1}{k}i}$.
	\begin{lemma}\label{Fq*sum}\rm (\cite{TangEHLCD})
		For any integer $t$ and $\beta\in\mathbb{F}_{q}^{*}$, we have
		$$\sum\limits_{i=1}^{s}\left(\beta\alpha_{i}\right)^{t}=\begin{cases}
			\beta^{t}s,&\text{if}\ s\mid t;\\
			0,&\text{otherwise}.
		\end{cases}$$ 
	\end{lemma}
	
	\begin{lemma}\label{kq-1sttaking}\rm (\cite{LiangLCDGRL}, Lemma 3.1)
		If $q-1$, $k$, $s$ and $t$ satisfy $\frac{q-1}{k}\nmid s-t$ and
		$$v_{2}\left(s-t\right)\neq v_{2}(q-1)-v_{2}(k)-1,$$  then any two components of the vector $\boldsymbol{\alpha}=\left(\gamma^{s}\alpha_{1},...,\gamma^{s}\alpha_{k},\gamma^{t}\alpha_{1},...,\gamma^{t}\alpha_{k}\right)$ are distinct over $\mathbb{F}_{q}$ and $\gamma^{sk}+\gamma^{t k}\in\mathbb{F}_{q}^{*}$. 
	\end{lemma}	
	
	\begin{lemma}\label{kq-1mutaking}
		If $\gcd(3k,q)=1$, $q-1\notin\left\{k,2k,3k\right\}$ and the vector
		$$\boldsymbol{\alpha}=\left(\gamma^{\mu}\alpha_{1},\ldots,\gamma^{\mu}\alpha_{k},\gamma^{\mu+1}\alpha_{1},\ldots,\gamma^{\mu+1}\alpha_{k},\gamma^{\mu+2}\alpha_{1},\ldots,\gamma^{\mu+2}\alpha_{k}\right),1\leq\mu\leq q-1,$$
		then any two components of $\boldsymbol{\alpha}$ are distinct over $\mathbb{F}_{q}$ and $k\sum\limits_{i=0}^{2}\gamma^{(\mu+i)k}\in\mathbb{F}_{q}^{*}$. 
	\end{lemma}	
	{\bf Proof.} We only need to prove that the following two statements are true,
	
	$(1)$ for any $1\leq i\neq j\leq k\mid q-1$, $\gamma^{\mu}\alpha_{i}\ne \gamma^{\mu+1}\alpha_{j},$ $\gamma^{\mu}\alpha_{i}\ne \gamma^{\mu+2}\alpha_{j},$ and  $\gamma^{\mu+1}\alpha_{i}\ne \gamma^{\mu+2}\alpha_{j}$;
	
	$(2)$ $k\sum\limits_{i=0}^{2}\gamma^{(\mu+i)k}\in\mathbb{F}_{q}^{*}.$
	
	{\bf For (1).} In the similar proof as that for Lemma $\ref{kq-1sttaking}$, we know that for any $1\leq i\neq j\leq k\mid q-1$, $$\gamma^{\mu}\alpha_{i}\ne \gamma^{\mu+1}\alpha_{j}\Longleftrightarrow\frac{q-1}{k}\nmid 1\Longleftrightarrow q-1\neq k,$$ $$\gamma^{\mu}\alpha_{i}\ne \gamma^{\mu+2}\alpha_{j}\Longleftrightarrow\frac{q-1}{k}\nmid 2\Longleftrightarrow q-1\neq k,2k,$$   $$\gamma^{\mu+1}\alpha_{i}\ne \gamma^{\mu+2}\alpha_{j}\Longleftrightarrow\frac{q-1}{k}\nmid 1\Longleftrightarrow q-1\neq k,$$
	i.e., the statement {\bf (1)} holds if and only if $q-1\notin \left\{k,2k\right\}$.
	
	{\bf For (2).} Note that $2\leq k\mid\mathrm{ord}\left(\gamma\right)=q-1$, and so $\gamma^{k}-1\neq 0$ if and only if $k\neq q-1.$ 
	By $\gcd(3k,q)=1$ and $$\left(\gamma^{k}-1\right)\sum\limits_{i=0}^{2}\gamma^{(\mu+i)k}=\gamma^{\mu k}\left(1+\gamma^{k}+\gamma^{2k}\right)\left(\gamma^{k}-1\right)=\gamma^{\mu k}\left(\gamma^{3k}-1\right),$$ we know that $k\sum\limits_{i=0}^{2}\gamma^{(\mu+i)k}\in \mathbb{F}_{q}^{*}$ if and only if $\gamma^{3k}-1\in \mathbb{F}_{q}^{*}$ and $k\neq q-1$,  i.e., $\mathrm{ord}\left(\gamma\right)=q-1\nmid 3k$ and $k\neq q-1$, namely, $\frac{q-1}{k}\nmid 3$ and $k\neq q-1$, it means that the statement {\bf (2)} holds if and only if $ q-1\notin\left\{k,3k\right\}.$
	
	From the above discussions, Lemma $\ref{kq-1mutaking}$ is immediately.
	
	$\hfill\Box$
	
	\subsection{The subset sum problem and Symmetric Polynomials}
	
	Fistly, we recall the definition of the subset sum problem over finite abelian groups.
	
	For an additively finite abelian group $G$ of size $n$, given $k\in\mathbb{Z}^{+}$, $x\in G$ and a subset $D\subseteq G$ with $|D|\geq k$, to determine the value of
	$$N\left(k,x,D\right)=\#\left\{T\subseteq D: \# T=k, \sum\limits_{t\in T}t=x\right\}$$
	is the subset sum problem over the finite abelian group $G$. 
	
	Especially, for $D=\mathbb{Z}_{n}$, Li and Wan \cite{abeliangroups} gave a formula of $N\left(k,x,D\right)$ as the following  
	\begin{lemma}\label{Znsubsetsum}{\rm( \cite{abeliangroups},Corollary 1.2)}
		Let $x\in \mathbb{Z}_{n}$, $k \in \mathbb{Z}$ with $1 \leqslant k \leqslant n$. Then we have
		$$N(k,x,\mathbb{Z}_{n}) = \frac{1}{n} \sum_{r\mid \gcd(n,k)}(-1)^{k+\frac{k}{r}}\binom{\frac{n}{s}}{\frac{k}{s}}\sum_{d\mid \gcd(r,x)} \mu\left(\frac{r}{d}\right)d.$$
	\end{lemma} 
	
	Secondly, we recall the definitions of elementary symmetric polynomials and complete symmetric polynomials.
	
	\begin{definition}\label{DofESpolynomial}{\rm (\cite{Ding2025nonGRS})}
		Let $n$ and $t$ be a positive integer.  
		
		$(1)$ The $t$-th elementary symmetric polynomial in $n$-variables is defined by
		\[
		\sigma_{t}(x_{1},x_{2},\dots,x_{n})
		=
		\begin{cases}
			1, & t=0;\\[4pt]
			\displaystyle\sum_{1\leq j_{1}<j_{2}<\cdots<j_{t}\leq n} x_{j_1}x_{j_2}\cdots x_{j_t}, & 1\leq t\leq n;\\[4pt]
			0, & t>n.
		\end{cases}
		\]
		And we denote $\sigma_{t} = \sigma_{t}(x_{1},x_{2},\dots,x_{n})$;
		
		$(2)$ The $t$-th complete symmetric polynomial in $n$-variables is defined by
		\[
		S_{t}(x_{1},x_{2},\dots,x_{n})
		=
		\begin{cases}
			0, & t<0;\\[4pt]
			\displaystyle\sum_{\substack{t_{1}+t_{2}+\cdots+t_{n}=t\\ t_{i}\geq 0}} x_{1}^{t_{1}}x_{2}^{t_{2}}\cdots x_{n}^{t_{n}}, & t\geq 0.
		\end{cases}
		\]
		And we denote $S_{t} = S_{t}(x_{1},x_{2},\dots,x_{n})$. 
	\end{definition}
	
	The following Lemma \ref{sigmaS} shows that the two types of symmetric polynomials defined in Definitions \ref{DofESpolynomial} are closely related.
	\begin{lemma}\label{sigmaS}{\rm( \cite{Symmetric}, Page 21, Eq. 2.6$^{\prime}$)} 
		Let $\sigma_{t}$ and $S_{t}$ be the elementary symmetric polynomial and the complete symmetric polynomial, respectively. Then for any $\ N\geq 1$, 
		$$\sum\limits_{t=0}^{N}(-1)^{t}\sigma_{t}S_{N-t}=0.$$
	\end{lemma}
	
	Finally, we present a lemma concerning complete symmetric polynomials, which will be used in Section \ref{sec6}.
	\begin{lemma}\label{usapowersum}{\rm( \cite{Ding2025nonGRS}, Lemma 2.6)}
		For any subset $\left\{\alpha_{1}, \ldots, \alpha_{n}\right\}\subseteq\mathbb{F}_{q}$ with $n\geq 3$, we have
		$$\sum\limits_{i=1}^{n}u_{i}\alpha_{i}^{h}=\begin{cases}
			0,&\text{if}\ 0\leq h\leq n-2;\\
			S_{h-n+1}(\alpha_{1},\ldots,\alpha_{n}),&\text{if}\ h\geq n-1,\\
		\end{cases}$$
		where $u_{i}=\prod\limits_{j=1, j \neq i}^{n}\left(\alpha_{i}-\alpha_{j}\right)^{-1}$ for $1 \leq i \leq n$. 
	\end{lemma}
	\section{Two classes of MDS or NMDS MGRS codes and their weight distributions}\label{sec3}
	In this section, we firstly introduce the definition of MGRS codes and the necessary and sufficient conditions for two classes of MGRS codes to be MDS. Then, we prove that two classes of MGRS codes are either MDS or NMDS codes. Furthermore, we derive the necessary and sufficient conditions for these codes to be NMDS, and describe the construction process in details. Finally, by using the subset sum problem, we completely determine the weight distributions for a special class of NMDS MGRS codes.

	\subsection{The definition of MGRS codes and the MDS property} 
	\begin{definition}\label{MGRSdefinition}{\rm( \cite{LiuMGRS}, Definition 3)}
		Let $2\leq k\leq n\leq q$, $1\leq t\leq k-1$, $\boldsymbol{\alpha}=\left(\alpha_{1}, \ldots, \alpha_{n}\right) \in \mathbb{F}_{q}^{n}$ with $\alpha_{i} \neq \alpha_{j}(i \neq j)$ and $\boldsymbol{v}=$ $\left(v_{1}, \ldots, v_{n}\right) \in\left(\mathbb{F}_{q}^{*}\right)^{n}$. The modified generalized Reed-Solomon (in short, MGRS) code $\mathrm{MGRS}_{n,k}\left(\boldsymbol{\alpha},\boldsymbol{v},\eta,t\right)$ is defined as
		$$
		\mathrm{MGRS}_{n,k}\left(\boldsymbol{\alpha},\boldsymbol{v},\eta,t\right)\triangleq\left\{\left(v_{1} f\left(\alpha_{1}\right), \ldots, v_{n-1} f\left(\alpha_{n-1}\right),v_{n}\left(f\left(0\right)+\eta f_{t}\right)\right) | f(x) \in \mathbb{F}_{q}^{k}[x]\right\},
		$$
		where $v_{n+1}\in\mathbb{F}_{q}^{*}$, $\eta\in\mathbb{F}_{q}$ and $f_t$ denotes the coefficient of $x^{t}$ in $f(x)$.
	\end{definition}
	\begin{remark}\label{MGRSgeneratormatrix}
		It is easy to know that $\mathrm{MGRS}_{n,k}\left(\boldsymbol{\alpha},\boldsymbol{v},\eta,t\right)$ has the following generator matrix 
		$$\small\boldsymbol{G}_{\boldsymbol{v},n-1,t} =\begin{pmatrix}
			v_{1}&v_{2}&\cdots&v_{n-1}&v_{n}\\
			v_{1}\alpha_{1}&v_{2}\alpha_{2}&\cdots&v_{n-1}\alpha_{n-1}&0\\
			
			\vdots&\vdots& &\vdots&\vdots\\
			v_{1}\alpha_{1}^{t-1}&v_{2}\alpha_{2}^{t-1}&\cdots&v_{n-1}\alpha_{n-1}^{t-1}&0\\
			
			v_{1}\alpha_{1}^{t}&v_{2}\alpha_{2}^{t}&\cdots&v_{n-1}\alpha_{n-1}^{t}&v_{n}\eta\\
			
			v_{1}\alpha_{1}^{t+1}&v_{2}\alpha_{2}^{t+1}&\cdots&v_{n-1}\alpha_{n-1}^{t+1}&0\\
			\vdots&\vdots& &\vdots&\vdots\\ 
			v_{1}\alpha_{1}^{k-1}&v_{2}\alpha_{2}^{k-1}&\cdots&v_{n-1}\alpha_{n-1}^{k-1}&0
		\end{pmatrix}_{k\times n}.$$
	\end{remark}  
	
	In the following Lemma \ref{NMDSMDSMGRS}, a sufficient and necessary condition for $\mathrm{MGRS}_{n,k}\left(\boldsymbol{\alpha},\boldsymbol{v},\eta,t\right)(t\in\left\{1,k-1\right\})$ to
	be MDS is given in the following
	\begin{lemma}\label{NMDSMDSMGRS}{\rm( \cite{LiuMGRS}, Corollaries 1-2)}
		For $\mathrm{MGRS}_{n,k}\left(\boldsymbol{\alpha},\boldsymbol{v},\eta,t\right)$ or $\mathrm{EMGRS}_{n,k}\left(\boldsymbol{\alpha},\boldsymbol{v},\eta,t\right)(t\in\left\{1,k-1\right\})$, the following four statements are true,
		
		$(1)$ $\mathrm{MGRS}_{n,k}\left(\boldsymbol{\alpha},\boldsymbol{v},\eta,1\right)$ is MDS if and only if for any $(k-1)$-element subset $\mathcal{L}\subseteq\left\{\alpha_{1},\ldots,\alpha_{n-1}\right\}$, $\frac{1}{\eta}\neq \sum\limits_{\alpha_{i}\in \mathcal{L}}\frac{1}{\alpha_{i}}$;
		
		$(2)$ $\mathrm{MGRS}_{n,k}\left(\boldsymbol{\alpha},\boldsymbol{v},\eta,k-1\right)$ is MDS if and only if for any $(k-1)$-element subset $\mathcal{K}\subseteq\left\{\alpha_{1},\ldots,\alpha_{n-1}\right\}$, $\eta\neq (-1)^{k}\prod\limits_{\alpha_{i}\in \mathcal{K}}\alpha_{i}$. 
	\end{lemma}
	
	\subsection{An equivalent condition for two classes of MGRS codes to be NMDS}\label{sec3.2}
	
	In this subsection, we focus on two classes of MGRS codes $\mathrm{MGRS}_{n,k}\left(\boldsymbol{\alpha},\boldsymbol{v},\eta,1\right)$ and $\mathrm{MGRS}_{n,k}\left(\boldsymbol{\alpha},\boldsymbol{v},\eta,k-1\right)$. In order to present a necessary and sufficient condition for these codes to be NMDS, we firstly present the following 
	\begin{proposition}\label{MDSNMDSMGRS1}
		If $\left\{\alpha_{1},\ldots,\alpha_{n-1}\right\}\subseteq\mathbb{F}_{q}^{*}$, then the code $\mathrm{MGRS}_{n,k}\left(\boldsymbol{\alpha},\boldsymbol{v},\eta,k-1\right)$ is either MDS or NMDS.
	\end{proposition}
	{\bf Proof.} By Definition \ref{MGRSdefinition} and the definition of the monomially equivalent, we only focus on $\mathrm{MGRS}_{n,k}\left(\boldsymbol{\alpha},\boldsymbol{1},\eta,k-1\right)$ which has the generator matrix
	$$\boldsymbol{G}_{\boldsymbol{1}
		,n-1,k-1} =\begin{pmatrix}
		1&1&\cdots&1&1\\
		\alpha_{1}&\alpha_{2}&\cdots&\alpha_{n-1}&0\\
		\alpha_{1}^{2}&\alpha_{2}^{2}&\cdots&\alpha_{n-1}^{2}&0\\
		\vdots&\vdots& &\vdots&\vdots\\
		\alpha_{1}^{k-2}&\alpha_{2}^{k-2}&\cdots&\alpha_{n-1}^{k-2}&0\\
		\alpha_{1}^{k-1}&\alpha_{2}^{k-1}&\cdots&\alpha_{n-1}^{k-1}&\eta
	\end{pmatrix}.
	$$
	
	If the code $\mathrm{MGRS}_{n,k}\left(\boldsymbol{\alpha},\boldsymbol{v},\eta,k-1\right)$ is not MDS, then by Lemma \ref{LCDMDSorNMDSequivalent}, there exist $k$ columns in $\boldsymbol{G}_{\boldsymbol{1}
		,n-1,k-1}$ that are  $\mathbb{F}_{q}$-linearly dependent, i.e., Lemma \ref{LCDMDSorNMDSequivalent} $(2)$ holds.
	
	Next, we prove Lemma \ref{LCDMDSorNMDSequivalent} $(1)$ and $(3)$.
	
	For Lemma \ref{LCDMDSorNMDSequivalent} $(1)$, we only prove that there exists some $(k-1)\times (k-1)$ non-zero minor for the following matrix 
	$$\boldsymbol{G}_{\boldsymbol{1}
		,n-1,k-1}^{(k-1)} =\begin{pmatrix}	1&1&\cdots&1&1\\
		\alpha_{1}&\alpha_{2}&\cdots&\alpha_{k-2}&0\\
		\alpha_{1}^{2}&\alpha_{2}^{2}&\cdots&\alpha_{k-2}^{2}&0\\
		\vdots&\vdots& &\vdots&\vdots\\
		\alpha_{1}^{k-2}&\alpha_{2}^{k-2}&\cdots&\alpha_{k-2}^{k-2}&0\\
		\alpha_{1}^{k-1}&\alpha_{2}^{k-1}&\cdots&\alpha_{k-2}^{k-1}&\eta
	\end{pmatrix}_{k\times (k-1)}.
	$$
	Now, we consider the submatrix $\boldsymbol{R}_1$ given by deleting the last row of $\boldsymbol{G}_{\boldsymbol{1}
		,n-1,k-1}^{(k-1)} $, i.e.,
	$$\boldsymbol{R}_1=\begin{pmatrix}
		1&1&\cdots&1&1\\
		\alpha_{1}&\alpha_{2}&\cdots&\alpha_{k-2}&0\\
		\alpha_{1}^{2}&\alpha_{2}^{2}&\cdots&\alpha_{k-2}^{2}&0\\
		\vdots&\vdots& &\vdots&\vdots\\
		\alpha_{1}^{k-2}&\alpha_{2}^{k-2}&\cdots&\alpha_{k-2}^{k-2}&0
	\end{pmatrix}_{(k-1)\times (k-1)},
	$$
	by $\left\{\alpha_{1},\ldots,\alpha_{n-1}\right\}\subseteq\mathbb{F}_{q}^{*}$, we have
	$$\det\left(\boldsymbol{R}_1\right)=(-1)^{k}\cdot\prod\limits_{i=1}^{k-2}\alpha_{i}\cdot\prod\limits_{1 \leq s <t\leq k-2}\left(\alpha _{t}-\alpha _{s}\right)\neq 0,$$
	which means that $\boldsymbol{R}_1$ is a $(k-1)\times (k-1)$ non-zero minor of $\boldsymbol{G}_{\boldsymbol{1}
		,n-1,k-1}^{(k-1)} $.
	
	For Lemma \ref{LCDMDSorNMDSequivalent} $(2)$, we only prove that there exists a $k\times k$ non-zero minor for the following matrix
	$$\boldsymbol{G}_{\boldsymbol{1},n-1,k-1}^{(k)} =\begin{pmatrix}
		1&1&\cdots&1&1\\
		\alpha_{1}&\alpha_{2}&\cdots&\alpha_{k}&0\\
		\alpha_{1}^{2}&\alpha_{2}^{2}&\cdots&\alpha_{k}^{2}&0\\
		\vdots&\vdots& &\vdots&\vdots\\
		\alpha_{1}^{k-2}&\alpha_{2}^{k-2}&\cdots&\alpha_{k}^{k-2}&0\\
		\alpha_{1}^{k-1}&\alpha_{2}^{k-1}&\cdots&\alpha_{k}^{k-1}&\eta
	\end{pmatrix}_{k\times (k+1)}.
	$$
	It's easy to prove that the submatrix given by deleting the last column of $\boldsymbol{G}_{\boldsymbol{1},n-1,k-1}^{(k)} $ is a $k\times k$ non-zero minor of $\boldsymbol{G}_{\boldsymbol{1},n-1,k-1}^{(k)} $.   
	
	This completes the proof of Theorem \ref{MDSNMDSMGRS1}.
	
	$\hfill\Box$
	
	\iffalse$$\boldsymbol{G}_{2}=\begin{pmatrix}
		1&1&\cdots&1&1\\
		\alpha_{1}&\alpha_{2}&\cdots&\alpha_{n-1}&\eta\\
		\alpha_{1}^{2}&\alpha_{2}^{2}&\cdots&\alpha_{n-1}^{2}&0\\
		\vdots&\vdots& &\vdots&\vdots\\
		\alpha_{1}^{k-2}&\alpha_{2}^{k-2}&\cdots&\alpha_{n-1}^{k-2}&0\\
		\alpha_{1}^{k-1}&\alpha_{2}^{k-1}&\cdots&\alpha_{n-1}^{k-1}&0
	\end{pmatrix}
	$$\fi
	
	Furthermore, by Lemma \ref{NMDSMDSEMGRS} (1) and Proposition \ref{MDSNMDSMGRS1}, we have the following 
	\begin{theorem}\label{NMDSMGRSK-1}
		If $\left\{\alpha_{1},\ldots,\alpha_{n-1}\right\}\subseteq\mathbb{F}_{q}^{*}$, then $\mathrm{MGRS}_{n,k}\left(\boldsymbol{\alpha},\boldsymbol{v},\eta,k-1\right)$ is NMDS if and only if there exists an  $(k-1)$-element subset $\mathcal{K}\subseteq\left\{\alpha_{1},\ldots,\alpha_{n-1}\right\}$ such that $\eta=(-1)^{k}\prod\limits_{\alpha_{i}\in \mathcal{K}}\alpha_{i}$.
	\end{theorem} 
	
	Based on the above Theorem \ref{NMDSMGRSK-1}, a method for constructing MDS MGRS codes with parameters $[n,k,n-k+1]$ or  constructing NMDS MGRS codes with parameters $[n,k,n-k]$ can be divided into the following four steps:
	
	{\bf Step 1.} Given a positive integer $n\leq q$ and select an $(n-1)$-element subset $$H=\left\{\alpha_{1},\ldots,\alpha_{n-1}\right\}\subseteq\mathbb{F}_{q}^{*};$$ 
	
	{\bf Step 2.} Given a positive integer $k$ with $2\leq k\leq n$ and calculate the following set
	$$\widetilde{H}_{1}=\left\{(-1)^{k}\prod\limits_{\alpha_{i}\in \mathcal{H}}\alpha_{i}:  \mathcal{H}\subseteq H, \#\mathcal{H}=k-1\right\};$$
	
	{\bf Step 3.} If $\widetilde{H}_{1}=\mathbb{F}_{q}^{*}$, it means $\eta\in\widetilde{H}_{1}$, i.e., there exists a  $(k-1)$-element subset $\mathcal{K}\subseteq\left\{\alpha_{1},\ldots,\alpha_{n-1}\right\}$ such that $\eta=(-1)^{k}\prod\limits_{\alpha_{i}\in \mathcal{K}}\alpha_{i}$, then $\mathrm{MGRS}_{n,k}\left(H,\boldsymbol{v},\eta,k-1\right)$ is NMDS;
	
	{\bf Step 4.} Otherwise, choice $\eta\in\mathbb{F}_{q}^{*}\backslash \widetilde{H}_{1}$, then $\mathrm{MGRS}_{n,k}\left(H,\boldsymbol{v},\eta,k-1\right)$ is MDS.
	
	In the similar proof as that for Proposition \ref{MDSNMDSMGRS1}, one can prove the following
	\begin{proposition}\label{MDSNMDSMGRS2}
		If $\left\{\alpha_{1},\ldots,\alpha_{n-1}\right\}\subseteq\mathbb{F}_{q}^{*}$, then $\mathrm{MGRS}_{n,k}\left(\boldsymbol{\alpha},\boldsymbol{v},\eta,1\right)$ is either MDS or NMDS.
	\end{proposition}
	
	Furthermore, by Lemma \ref{NMDSMDSEMGRS} (2) and Proposition \ref{MDSNMDSMGRS2}, we have the following
	\begin{theorem}\label{NMDSMGRS1}
		If $\left\{\alpha_{1},\ldots,\alpha_{n-1}\right\}\subseteq\mathbb{F}_{q}^{*}$, then $\mathrm{MGRS}_{n,k}\left(\boldsymbol{\alpha},\boldsymbol{v},\eta,1\right)$ is NMDS if and only if there exists a $(k-1)$-element subset $\mathcal{L}\subseteq\left\{\alpha_{1},\ldots,\alpha_{n-1}\right\}$ such that $\frac{1}{\eta}=\sum\limits_{\alpha_{i}\in \mathcal{L}}\frac{1}{\alpha_{i}}$.
	\end{theorem}
	
	Based on the above Theorem \ref{NMDSMGRS1}, a method for constructing MDS MGRS codes with parameters $[n,k,n-k+1]$ or NMDS MGRS codes with parameters $[n,k,n-k]$ can be divided into the following four steps:
	
	{\bf Step 1.} Given a positive integer $n\leq q$ and select an $(n-1)$-element subset $$H=\left\{\alpha_{1},\ldots,\alpha_{n-1}\right\}\subseteq\mathbb{F}_{q}^{*};$$ 
	
	{\bf Step 2.} Given a positive integer $k$ with $2\leq k\leq n$ and calculate the following set
	$$\widetilde{H}_{2}=\left\{\sum\limits_{\alpha_{i}\in \mathcal{H}}\alpha_{i}^{-1}:  \mathcal{H}\subseteq H, \#\mathcal{H}=k-1\right\};$$
	
	{\bf Step 3.} If $\widetilde{H}_{2}=\mathbb{F}_{q}$, it means $\eta^{-1}\in\widetilde{H}_{2}$, i.e., there exist a  $(k-1)$-element subset $\mathcal{L}\subseteq\left\{\alpha_{1},\ldots,\alpha_{n-1}\right\}$ such that $\eta^{-1}=\sum\limits_{\alpha_{i}\in \mathcal{L}}\alpha_{i}$, then $\mathrm{MGRS}_{n,k}\left(H,\boldsymbol{v},\eta,1\right)$ is NMDS;
	
	{\bf Step 4.} Otherwise, choice $\eta^{-1}\in\mathbb{F}_{q}\backslash \widetilde{H}_{2}$, then $\mathrm{MGRS}_{n,k}\left(H,\boldsymbol{v},\eta,1\right)$ is MDS.
	\subsection{The weight distributions for a class of NMDS MGRS codes}\label{sec3.3} 
	In this subsection, by the counting formula of the subset sum over $\mathbb{Z}_{n}$, for a special class of NMDS MGRS codes discussed in Section \ref{sec3.2}, we completely determine the weight distribution.
	\begin{theorem}\label{MGRSk-1weight}
		If $\mathbb{F}_{q}^{*}=\langle\omega\rangle$ and $\eta=\omega^{t},t\in\left\{1,2,\ldots,q-1\right\}$, then for the NMDS MGRS code $\mathrm{MGRS}_{n,k}\left(\mathbb{F}_{q}^{*},\boldsymbol{v},\eta,k-1\right)$, the total
		number $A_{n-k}$ of codewords with Hamming weight $n-k$ is 
		$$A_{n-k}=\sum_{r\mid \gcd(q-1,k-1)}(-1)^{k-1+\frac{k-1}{r}}\binom{\frac{q-1}{r}}{\frac{k-1}{r}}\sum_{\substack{d\mid (r,t-\frac{k(q-1)}{2})}} \mu\left(\frac{r}{d}\right)d.$$
	\end{theorem}
	{\bf Proof.} By Definition \ref{MGRSdefinition}, it is easy to get  $$\mathrm{MGRS}_{n,k}\left(\boldsymbol{\alpha},\boldsymbol{v},\eta,k-1\right)\triangleq\left\{\left(v_{1} f\left(\alpha_{1}\right), \ldots, v_{n-1} f\left(\alpha_{n-1}\right),v_{n}\left(f\left(0\right)+\eta f_{k-1}\right)\right) | f(x) \in \mathbb{F}_{q}^{k}[x]\right\}.$$
	\iffalse furthermore, the MGRS code $\mathrm{MGRS}_{n,k}\left(\boldsymbol{\alpha},\boldsymbol{v},\eta,k-1\right)$ is NMDS if and only if $$d\left(\mathrm{MGRS}_{n,k}\left(\boldsymbol{\alpha},\boldsymbol{v},\eta,k-1\right)\right)=n-k.$$\fi
	Next, we calculate the number of codewords with minimum
	weight $n-k$. For the convenience, we set $\mathcal{I}=\left\{\alpha_{1},\ldots,\alpha_{n-1}\right\}$. Then $\wt\left(\boldsymbol{c}\right)=n-k$ if and only if 
	$$\#\left\{\alpha_{i}\in\mathcal{I}\ |\ f(\alpha_{i})=0\right\}=k\ \text{and}\ f\left(0\right)+\eta f_{k-1}\neq 0$$
	or
	$$\#\left\{\alpha_{i}\in\mathcal{I}\ |\ f(\alpha_{i})=0\right\}=k-1\ \text{and}\ f\left(0\right)+\eta f_{k-1}=0.$$ 
	
	{\bf Case 1.} If $f\left(0\right)+\eta f_{k-1}\neq0$, then 
	$$
	A_{n-k}=\#\left\{\mathbb{F}_{q}^{k}[x]\cap\left\{f(x)=\lambda\prod\limits_{\alpha_{i}\in\mathcal{J}}\left(x-\alpha_{i}\right)|\ \#\mathcal{J}=k,\mathcal{J}\subseteq\mathcal{I},\lambda\in\mathbb{F}_{q}^{*},f\left(0\right)+\eta f_{k-1}\neq 0\right\}\right\}.$$ 
	Note that for any $f(x)\in \mathbb{F}_{q}^{k}[x]$, $\deg\left(f(x)\right)\leq k-1$, and then $A_{n-k}=0.$
	
	{\bf Cases 2.} If $f\left(0\right)+\eta f_{k-1}=0$, then 
	$$\begin{aligned} 
		A_{n-k}=&\#\left\{\mathbb{F}_{q}^{k}[x]\cap\left\{f(x)=\lambda\prod\limits_{\alpha_{i}\in\mathcal{J}}\left(x-\alpha_{i}\right)|\ \#\mathcal{J}=k-1,\mathcal{J}\subseteq\mathcal{I},\lambda\in\mathbb{F}_{q}^{*},f\left(0\right)+\eta f_{k-1}= 0\right\}\right\}\\
		=&\#\left\{\mathbb{F}_{q}^{k}[x]\cap\left\{f(x)=\lambda\prod\limits_{\alpha_{i}\in\mathcal{J}}\left(x-\alpha_{i}\right)|\ \#\mathcal{J}=k-1,\mathcal{J}\subseteq\mathcal{I},\lambda\in\mathbb{F}_{q}^{*},f_{k-1}=-\frac{f\left(0\right)}{\eta}\right\}\right\}.
	\end{aligned}$$
	Note that
	$$f(x)=\lambda\prod\limits_{\alpha_{i}\in\mathcal{J}}\left(x-\alpha_{i}\right)=\lambda\left(x^{k-1}-\sum\limits_{\alpha_{i}\in\mathcal{J}}\alpha_{i}x^{k-2}+\cdots+(-1)^{k-1}\prod\limits_{\alpha_{i}\in\mathcal{J}}\alpha_{i}\right),$$
	thus, $f_{k-1}=\lambda$ and $f(0)=(-1)^{k-1}\cdot\lambda\cdot\prod\limits_{\alpha_{i}\in\mathcal{J}}\alpha_{i}$. Furthermore, we have
	$$\begin{aligned}A_{n-k}=&\#\left\{\mathbb{F}_{q}^{k}[x]\cap\left\{f(x)=\lambda\prod\limits_{\alpha_{i}\in\mathcal{J}}\left(x-\alpha_{i}\right)|\ \#\mathcal{J}=k-1,\mathcal{J}\subseteq\mathcal{I},\lambda\in\mathbb{F}_{q}^{*},\eta=(-1)^{k}\prod\limits_{\alpha_{i}\in\mathcal{J}}\alpha_{i}\right\}\right\}\\ 
		=&\# \left\{\left(\lambda,\mathcal{J}\right)\ |\ \eta=(-1)^{k}\prod\limits_{\alpha_{i}\in\mathcal{J}}\alpha_{i},\lambda\in\mathbb{F}_{q}^{*},\mathcal{J}\subseteq\mathcal{I}\ \text{and}\ \#\mathcal{J}=k-1\right\}\\
		=&(q-1)\# \left\{\mathcal{J}\ |\ \eta=(-1)^{k}\prod\limits_{\alpha_{i}\in\mathcal{J}}\alpha_{i},\mathcal{J}\subseteq\mathcal{I}\ \text{and}\ \#\mathcal{J}=k-1\right\}.
	\end{aligned}$$ 
	
	Now, by $\eta=\omega^{t},t\in\left\{1,2,\ldots,q-1\right\}$, we can get  
	$$\begin{aligned}A_{n-k}=&(q-1)\# \left\{\mathcal{L}\ |\ \omega^{t}=\omega^{\frac{k(q-1)}{2}+\sum\limits_{{i}\in\mathcal{L}}i},\mathcal{L}\subseteq\left\{1,2,\ldots,q-1\right\}\ \text{and}\ \#\mathcal{L}=k-1\right\}\\
		=&(q-1)\# \left\{\mathcal{L}\ |\frac{k(q-1)}{2}+\sum\limits_{{i}\in\mathcal{L}}i\equiv t(\bmod\ q-1),\mathcal{L}\subseteq\left\{1,2,\ldots,q-1\right\}\ \text{and}\ \#\mathcal{L}=k-1\right\}\\
		=&(q-1)\# \left\{\mathcal{L}\ |\sum\limits_{{i}\in\mathcal{L}}i\equiv t-\frac{k(q-1)}{2}(\bmod\ q-1),\mathcal{L}\subseteq\left\{1,2,\ldots,q-1\right\}\ \text{and}\ \#\mathcal{L}=k-1\right\}\\
		=&(q-1)\# \left\{\mathcal{L}\ |\sum\limits_{{i}\in\mathcal{L}}i= t-\frac{k(q-1)}{2},\mathcal{L}\subseteq\mathbb{Z}_{q-1}\ \text{and}\ \#\mathcal{L}=k-1\right\}.
	\end{aligned}$$
	And so by Lemma \ref{Znsubsetsum}, we have
	$$\begin{aligned}
		A_{n-k}=&(q-1)N\left(k-1,t-\frac{k(q-1)}{2},\mathbb{Z}_{q-1}\right)\\
		=&\sum_{r\mid \gcd(q-1,k-1)}(-1)^{k-1+\frac{k-1}{r}}\binom{\frac{q-1}{r}}{\frac{k-1}{r}}\sum_{\substack{d\mid (r,t-\frac{k(q-1)}{2})}} \mu\left(\frac{r}{d}\right)d .
	\end{aligned}$$   
	
	This completes the proof of Theorem \ref{MGRSk-1weight}.
	
	$\hfill\Box$

	\section{Two classes of MDS or NMDS EMGRS codes and their weight distributions}\label{sec4}
	In this section, we firstly introduce the definition of EMGRS codes and the necessary and sufficient conditions for two classes of EMGRS codes to be MDS. Then, we prove that two classes of EMGRS codes are either MDS or NMDS codes. Furthermore, we derive the necessary and sufficient conditions for these codes to be NMDS, and describe the construction process in details. Finally, by using the subset sum problem, we completely determine the weight distributions for a special class of NMDS EMGRS codes.
	
	\subsection{The definition of EMGRS codes and the MDS property} 
	The definition of the EMGRS code is given in the following 
	\begin{definition}\label{EMGRSdefinition}{\rm( \cite{LiuMGRS}, Definition 4)}
		Let $2\leq k\leq n\leq q$, $1\leq t\leq k-1$, $\boldsymbol{\alpha}=\left(\alpha_{1}, \ldots, \alpha_{n}\right) \in \mathbb{F}_{q}^{n}$ with $\alpha_{i} \neq \alpha_{j}(i \neq j)$ and $\boldsymbol{v}=$ $\left(v_{1}, \ldots, v_{n}\right) \in\left(\mathbb{F}_{q}^{*}\right)^{n}$. The extended modified generalized Reed-Solomon (in short, EMGRS) code $\mathrm{EMGRS}_{n+1,k}\left(\boldsymbol{\alpha},\boldsymbol{v},\eta,t\right)$ is defined as
		$$
		\mathrm{EMGRS}_{n,k}\left(\boldsymbol{\alpha},\boldsymbol{v},\eta,t\right)\triangleq\left\{\left(v_{1} f\left(\alpha_{1}\right), \ldots, v_{n-1} f\left(\alpha_{n-1}\right),v_{n}\left(f\left(0\right)+\eta f_{t}\right),v_{n+1}f_{k-1}\right) | f(x) \in \mathbb{F}_{q}^{k}[x]\right\},
		$$
		where $v_{n+1}\in\mathbb{F}_{q}^{*}$, $\eta\in\mathbb{F}_{q}$ and $f_t$ denotes the coefficient of $x^{t}$ in $f(x)$.
	\end{definition}
	\iffalse\begin{remark}
		It is easy to know that the MGRS code $\mathrm{MGRS}_{n,k}\left(\boldsymbol{\alpha},\boldsymbol{v},\eta,t\right)$ has the following generator matrix 
		$$\boldsymbol{G}_{\boldsymbol{v},n-1,t} =\begin{pmatrix}
			v_{1}&v_{2}&\cdots&v_{n-1}&v_{n}\\
			v_{1}\alpha_{1}&v_{2}\alpha_{2}&\cdots&v_{n-1}\alpha_{n-1}&0\\
			
			\vdots&\vdots& &\vdots&\vdots\\
			v_{1}\alpha_{1}^{t-1}&v_{2}\alpha_{2}^{t-1}&\cdots&v_{n-1}\alpha_{n-1}^{t-1}&0\\
			
			v_{1}\alpha_{1}^{t}&v_{2}\alpha_{2}^{t}&\cdots&v_{n-1}\alpha_{n-1}^{t}&v_{n}\eta\\
			
			v_{1}\alpha_{1}^{t+1}&v_{2}\alpha_{2}^{t+1}&\cdots&v_{n-1}\alpha_{n-1}^{t+1}&0\\
			\vdots&\vdots& &\vdots&\vdots\\ 
			v_{1}\alpha_{1}^{k-1}&v_{2}\alpha_{2}^{k-1}&\cdots&v_{n-1}\alpha_{n-1}^{k-1}&0
		\end{pmatrix}_{k\times n}.$$
	\end{remark}\fi 
	
	In the following Lemma \ref{NMDSMDSEMGRS}, a sufficient and necessary condition for $\mathrm{MGRS}_{n,k}\left(\boldsymbol{\alpha},\boldsymbol{v},\eta,t\right)$ or $\mathrm{EMGRS}_{n,k}\left(\boldsymbol{\alpha},\boldsymbol{v},\eta,t\right)(t\in\left\{1,k-1\right\})$  to
	be MDS are given, respectively.
	\begin{lemma}\label{NMDSMDSEMGRS}{\rm( \cite{LiuMGRS}, Corollaries 1-4)}
		For $\mathrm{EMGRS}_{n,k}\left(\boldsymbol{\alpha},\boldsymbol{v},\eta,t\right)(t\in\left\{1,k-1\right\})$, the following two statements are true, 
		
		$(1)$ $\mathrm{EMGRS}_{n,k}\left(\boldsymbol{\alpha},\boldsymbol{v},\eta,1\right)$ is MDS if and only if $\frac{1}{\eta}\neq \sum\limits_{s=1}^{m}\alpha_{i_{s}}^{-1}$, where $\left\{i_1,\ldots,i_m\right\}$ is any $m$-element subset of $\left\{1,2,\ldots,n-1\right\}$ and $m\in\left\{k-1,k-2\right\}$; 
		
		$(2)$ if $\left\{\alpha_{1},\ldots,\alpha_{n-1}\right\}\subseteq\mathbb{F}_{q}^{*}$, then $\mathrm{EMGRS}_{n+1,k}\left(\boldsymbol{\alpha},\boldsymbol{v},\eta,k-1\right)$ is MDS if and only if $\mathrm{MGRS}_{n,k}\left(\boldsymbol{\alpha},\boldsymbol{v},\eta,k-1\right)$ is MDS.
	\end{lemma}
	
	\subsection{An equivalent condition for two classes of EMGRS codes to be NMDS}\label{sec4.2}
	In this subsection, we focus on two classes of EMGRS codes $\mathrm{EMGRS}_{n,k}\left(\boldsymbol{\alpha},\boldsymbol{v},\eta,1\right)$ and $\mathrm{EMGRS}_{n,k}\left(\boldsymbol{\alpha},\boldsymbol{v},\eta,k-1\right)$. In order to present the necessary and sufficient conditions for these codes to be NMDS, we firstly present the following 
	\begin{proposition}\label{MDSNMDSEMGRS1}
		If $\left\{\alpha_{1},\ldots,\alpha_{n-1}\right\}\subseteq\mathbb{F}_{q}^{*}$, then $\mathrm{EMGRS}_{n,k}\left(\boldsymbol{\alpha},\boldsymbol{v},\eta,k-1\right)$ is either MDS or NMDS.
	\end{proposition}
	{\bf Proof.} By Definition \ref{MGRSdefinition} and the definition of the monomially equivalent, we only focus on $\mathrm{EMGRS}_{n,k}\left(\boldsymbol{\alpha},\boldsymbol{1},\eta,k-1\right)$, which has the generator matrix
	$$\boldsymbol{G}_{\boldsymbol{1},n-1,k-1}^{(E)} =\begin{pmatrix}
		1&1&\cdots&1&1&0\\
		\alpha_{1}&\alpha_{2}&\cdots&\alpha_{n-1}&0&0\\
		\alpha_{1}^{2}&\alpha_{2}^{2}&\cdots&\alpha_{n-1}^{2}&0&0\\
		\vdots&\vdots& &\vdots&\vdots&\vdots\\
		\alpha_{1}^{k-2}&\alpha_{2}^{k-2}&\cdots&\alpha_{n-1}^{k-2}&0&0\\
		\alpha_{1}^{k-1}&\alpha_{2}^{k-1}&\cdots&\alpha_{n-1}^{k-1}&\eta&1
	\end{pmatrix}.
	$$
	
	If $\mathrm{EMGRS}_{n,k}\left(\boldsymbol{\alpha},\boldsymbol{v},\eta,k-1\right)$ is not MDS, then by Lemma \ref{LCDMDSorNMDSequivalent}, there exists $k$ linearly dependent
	columns in $\boldsymbol{G}_{\boldsymbol{1},n-1,k-1}^{(E)}$, i.e., Lemma \ref{LCDMDSorNMDSequivalent} $(2)$ holds.
	
	Next, we prove Lemma \ref{LCDMDSorNMDSequivalent} $(1)$ and $(3)$.
	
	For Lemma \ref{LCDMDSorNMDSequivalent} $(1)$, we only prove that there exists some $(k-1)\times (k-1)$ non-zero minor for the following matrix 
	$$\boldsymbol{S}_{1} =\begin{pmatrix}	1&1&\cdots&1\\
		\alpha_{1}&\alpha_{2}&\cdots&\alpha_{k-1}\\
		\alpha_{1}^{2}&\alpha_{2}^{2}&\cdots&\alpha_{k-1}^{2}\\
		\vdots&\vdots& &\vdots\\
		\alpha_{1}^{k-2}&\alpha_{2}^{k-2}&\cdots&\alpha_{k-1}^{k-2}\\
		\alpha_{1}^{k-1}&\alpha_{2}^{k-1}&\cdots&\alpha_{k-1}^{k-1}\end{pmatrix}_{k\times (k-1)},\ \boldsymbol{S}_{2} =\begin{pmatrix}	1&1&\cdots&1&1\\
		\alpha_{1}&\alpha_{2}&\cdots&\alpha_{k-2}&0\\
		\alpha_{1}^{2}&\alpha_{2}^{2}&\cdots&\alpha_{k-2}^{2}&0\\
		\vdots&\vdots& &\vdots&\vdots\\
		\alpha_{1}^{k-2}&\alpha_{2}^{k-2}&\cdots&\alpha_{k-2}^{k-2}&0\\
		\alpha_{1}^{k-1}&\alpha_{2}^{k-1}&\cdots&\alpha_{k-2}^{k-1}&\eta
	\end{pmatrix}_{k\times (k-1)},$$
	$$\boldsymbol{S}_{3} =\begin{pmatrix}	1&1&\cdots&1&0\\
		\alpha_{1}&\alpha_{2}&\cdots&\alpha_{k-2}&0\\
		\alpha_{1}^{2}&\alpha_{2}^{2}&\cdots&\alpha_{k-2}^{2}&0\\
		\vdots&\vdots& &\vdots&\vdots\\
		\alpha_{1}^{k-2}&\alpha_{2}^{k-2}&\cdots&\alpha_{k-2}^{k-2}&0\\
		\alpha_{1}^{k-1}&\alpha_{2}^{k-1}&\cdots&\alpha_{k-2}^{k-1}&1
	\end{pmatrix}_{k\times (k-1)},\ \boldsymbol{S}_{4} =\begin{pmatrix}	1&1&\cdots&1&1&0\\
		\alpha_{1}&\alpha_{2}&\cdots&\alpha_{k-3}&0&0\\
		\alpha_{1}^{2}&\alpha_{2}^{2}&\cdots&\alpha_{k-3}^{2}&0&0\\
		\vdots&\vdots& &\vdots&\vdots&\vdots\\
		\alpha_{1}^{k-2}&\alpha_{2}^{k-2}&\cdots&\alpha_{k-3}^{k-2}&0&0\\
		\alpha_{1}^{k-1}&\alpha_{2}^{k-1}&\cdots&\alpha_{k-3}^{k-1}&\eta&1
	\end{pmatrix}_{k\times (k-1)}.
	$$
	
	For $\boldsymbol{S}_{1} $, we consider the submatrix $\boldsymbol{T}_1$ given by deleting the last row of $\boldsymbol{S}_{1} $, i.e.,
	$$\boldsymbol{T}_1=\begin{pmatrix}	1&1&\cdots&1\\
		\alpha_{1}&\alpha_{2}&\cdots&\alpha_{k-1}\\
		\alpha_{1}^{2}&\alpha_{2}^{2}&\cdots&\alpha_{k-1}^{2}\\
		\vdots&\vdots& &\vdots\\
		\alpha_{1}^{k-2}&\alpha_{2}^{k-2}&\cdots&\alpha_{k-1}^{k-2}
	\end{pmatrix}_{(k-1)\times (k-1)},
	$$
	it is obvious that
	$$\det\left(\boldsymbol{T}_1\right)=\prod\limits_{1 \leq s <t\leq k-1}\left(\alpha _{t}-\alpha _{s}\right)\neq 0,$$
	which means that $\boldsymbol{T}_1$ is a $(k-1)\times (k-1)$ non-zero minor of $\boldsymbol{S}_{1} $.
	
	For $\boldsymbol{S}_{2} $, we consider the submatrix $\boldsymbol{T}_2$ given by deleting the last row of $\boldsymbol{S}_{2} $, i.e.,
	$$\boldsymbol{T}_2=\begin{pmatrix}	1&1&\cdots&1&1\\
		\alpha_{1}&\alpha_{2}&\cdots&\alpha_{k-2}&0\\
		\alpha_{1}^{2}&\alpha_{2}^{2}&\cdots&\alpha_{k-2}^{2}&0\\
		\vdots&\vdots& &\vdots&\vdots\\
		\alpha_{1}^{k-2}&\alpha_{2}^{k-2}&\cdots&\alpha_{k-2}^{k-2}&0
	\end{pmatrix}_{(k-1)\times (k-1)},
	$$
	by $\left\{\alpha_{1},\ldots,\alpha_{n-1}\right\}\subseteq\mathbb{F}_{q}^{*}$, we have
	$$\det\left(\boldsymbol{T}_2\right)=(-1)^{k}\cdot\prod\limits_{i=1}^{k-1}\alpha_{i}\cdot\prod\limits_{1 \leq s <t\leq k-2}\left(\alpha _{t}-\alpha _{s}\right)\neq 0,$$
	which means that $\boldsymbol{T}_2$ is a $(k-1)\times (k-1)$ non-zero minor of $\boldsymbol{S}_{2} $.
	
	For $\boldsymbol{S}_{3} $, we consider the submatrix $\boldsymbol{T}_3$ given by deleting the penultimate row of $\boldsymbol{S}_{3} $, i.e.,
	$$\boldsymbol{T}_3=\begin{pmatrix}	1&1&\cdots&1&0\\
		\alpha_{1}&\alpha_{2}&\cdots&\alpha_{k-2}&0\\
		\alpha_{1}^{2}&\alpha_{2}^{2}&\cdots&\alpha_{k-2}^{2}&0\\
		\vdots&\vdots& &\vdots&\vdots\\
		\alpha_{1}^{k-3}&\alpha_{2}^{k-3}&\cdots&\alpha_{k-2}^{k-3}&0\\
		\alpha_{1}^{k-1}&\alpha_{2}^{k-1}&\cdots&\alpha_{k-2}^{k-1}&1
	\end{pmatrix}_{(k-1)\times (k-1)},
	$$
	it is obvious that
	$$\det\left(\boldsymbol{T}_3\right)=(-1)^{2k-2}\cdot\prod\limits_{1 \leq s <t\leq k-2}\left(\alpha _{t}-\alpha _{s}\right)\neq 0,$$
	which means that $\boldsymbol{T}_3$ is a $(k-1)\times (k-1)$ non-zero minor of $\boldsymbol{S}_{3} $.
	
	For $\boldsymbol{S}_{4} $, we consider the submatrix $\boldsymbol{T}_4$ given by deleting the penultimate row of $\boldsymbol{S}_{4} $, i.e.,
	$$\boldsymbol{T}_4=\begin{pmatrix}	1&1&\cdots&1&1&0\\
		\alpha_{1}&\alpha_{2}&\cdots&\alpha_{k-3}&0&0\\
		\alpha_{1}^{2}&\alpha_{2}^{2}&\cdots&\alpha_{k-3}^{2}&0&0\\
		\vdots&\vdots& &\vdots&\vdots&\vdots\\
		\alpha_{1}^{k-3}&\alpha_{2}^{k-3}&\cdots&\alpha_{k-3}^{k-3}&0&0\\
		\alpha_{1}^{k-1}&\alpha_{2}^{k-1}&\cdots&\alpha_{k-3}^{k-1}&\eta&1
	\end{pmatrix}_{(k-1)\times (k-1)},
	$$
	by $\left\{\alpha_{1},\ldots,\alpha_{n-1}\right\}\subseteq\mathbb{F}_{q}^{*}$, we have
	$$\det\left(\boldsymbol{T}_4\right)=(-1)^{2k-2}\cdot(-1)^{k-1}\cdot\prod\limits_{i=1}^{k-3}\alpha_{i}\cdot\prod\limits_{1 \leq s <t\leq k-3}\left(\alpha _{t}-\alpha _{s}\right)\neq 0,$$
	which means that $\boldsymbol{T}_4$ is a $(k-1)\times (k-1)$ non-zero minor of $\boldsymbol{S}_{4} $.
	
	For Lemma \ref{LCDMDSorNMDSequivalent} $(2)$, we only prove that there exists some $k\times k$ non-zero minor for the following matrix
	$$\boldsymbol{U}_{1} =\begin{pmatrix}	1&1&\cdots&1\\
		\alpha_{1}&\alpha_{2}&\cdots&\alpha_{k+1}\\
		\alpha_{1}^{2}&\alpha_{2}^{2}&\cdots&\alpha_{k+1}^{2}\\
		\vdots&\vdots& &\vdots\\
		\alpha_{1}^{k-2}&\alpha_{2}^{k-2}&\cdots&\alpha_{k+1}^{k-2}\\
		\alpha_{1}^{k-1}&\alpha_{2}^{k-1}&\cdots&\alpha_{k+1}^{k-1}\end{pmatrix}_{k\times (k+1)},\ \boldsymbol{S}_{2} =\begin{pmatrix}	1&1&\cdots&1&1\\
		\alpha_{1}&\alpha_{2}&\cdots&\alpha_{k}&0\\
		\alpha_{1}^{2}&\alpha_{2}^{2}&\cdots&\alpha_{k}^{2}&0\\
		\vdots&\vdots& &\vdots&\vdots\\
		\alpha_{1}^{k-2}&\alpha_{2}^{k-2}&\cdots&\alpha_{k}^{k-2}&0\\
		\alpha_{1}^{k-1}&\alpha_{2}^{k-1}&\cdots&\alpha_{k}^{k-1}&\eta
	\end{pmatrix}_{k\times (k+1)},$$
	$$\boldsymbol{U}_{3} =\begin{pmatrix}	1&1&\cdots&1&0\\
		\alpha_{1}&\alpha_{2}&\cdots&\alpha_{k}&0\\
		\alpha_{1}^{2}&\alpha_{2}^{2}&\cdots&\alpha_{k}^{2}&0\\
		\vdots&\vdots& &\vdots&\vdots\\
		\alpha_{1}^{k-2}&\alpha_{2}^{k-2}&\cdots&\alpha_{k}^{k-2}&0\\
		\alpha_{1}^{k-1}&\alpha_{2}^{k-1}&\cdots&\alpha_{k}^{k-1}&1
	\end{pmatrix}_{k\times (k+1)},\ \boldsymbol{U}_{4} =\begin{pmatrix}	1&1&\cdots&1&1&0\\
		\alpha_{1}&\alpha_{2}&\cdots&\alpha_{k-1}&0&0\\
		\alpha_{1}^{2}&\alpha_{2}^{2}&\cdots&\alpha_{k-1}^{2}&0&0\\
		\vdots&\vdots& &\vdots&\vdots&\vdots\\
		\alpha_{1}^{k-2}&\alpha_{2}^{k-2}&\cdots&\alpha_{k-1}^{k-2}&0&0\\
		\alpha_{1}^{k-1}&\alpha_{2}^{k-1}&\cdots&\alpha_{k-1}^{k-1}&\eta&1
	\end{pmatrix}_{k\times (k+1)}.
	$$
	
	For $\boldsymbol{U}_{i}(i=1,2,3)$, by deleting the last column of $\boldsymbol{U}_{i}(i=1,2,3)$, we can obtain a  $k\times k$ Vandermonde determinant, which is obvious a $k\times k$ non-zero minor of $\boldsymbol{U}_{i}(i=1,2,3)$.
	
	For $\boldsymbol{U}_{4}$, by deleting the penultimate column of $\boldsymbol{U}_{4}$, we can obtain a  $k\times k$ non-zero minor of $\boldsymbol{U}_{4}$ as the following
	$$\begin{pmatrix}	1&1&\cdots&1&0\\
		\alpha_{1}&\alpha_{2}&\cdots&\alpha_{k-1}&0\\
		\alpha_{1}^{2}&\alpha_{2}^{2}&\cdots&\alpha_{k-1}^{2}&0\\
		\vdots&\vdots& &\vdots&\vdots\\
		\alpha_{1}^{k-2}&\alpha_{2}^{k-2}&\cdots&\alpha_{k-1}^{k-2}&0\\
		\alpha_{1}^{k-1}&\alpha_{2}^{k-1}&\cdots&\alpha_{k-1}^{k-1}&1
	\end{pmatrix}_{k\times k}.$$  
	
	This completes the proof of Proposition \ref{MDSNMDSEMGRS1}.
	
	$\hfill\Box$
	
	Furthermore, by Lemma \ref{NMDSMDSEMGRS} (4) and Proposition \ref{MDSNMDSEMGRS1}, we have the following 
	\begin{theorem}\label{NMDSEMGRSK-1}
		If $\left\{\alpha_{1},\ldots,\alpha_{n-1}\right\}\subseteq\mathbb{F}_{q}^{*}$, then $\mathrm{EMGRS}_{n,k}\left(\boldsymbol{\alpha},\boldsymbol{v},\eta,k-1\right)$ is NMDS if and only if there exists some $(k-1)$-element subset $\mathcal{K}\subseteq\left\{\alpha_{1},\ldots,\alpha_{n-1}\right\}$ such that $\eta=(-1)^{k}\prod\limits_{\alpha_{i}\in \mathcal{K}}\alpha_{i}$.
	\end{theorem} 
	
	In the similar proof as that for Proposition \ref{MDSNMDSEMGRS1}, one can prove the following
	\begin{theorem}\label{MDSNMDSEMGRS2}
		If $\left\{\alpha_{1},\ldots,\alpha_{n-1}\right\}\subseteq\mathbb{F}_{q}^{*}$, then $\mathrm{EMGRS}_{n,k}\left(\boldsymbol{\alpha},\boldsymbol{v},\eta,1\right)$ is either MDS or NMDS.
	\end{theorem}
	
	Furthermore, by Lemma \ref{NMDSMDSEMGRS} (3) and Theorem \ref{MDSNMDSEMGRS2}, we have the following
	\begin{theorem}\label{NMDSEMGRS1}
		If $\left\{\alpha_{1},\ldots,\alpha_{n-1}\right\}\subseteq\mathbb{F}_{q}^{*}$, then $\mathrm{EMGRS}_{n,k}\left(\boldsymbol{\alpha},\boldsymbol{v},\eta,1\right)$ is NMDS if and only if there exist some $m$-element subset $\left\{i_1,\ldots,i_m\right\}$ of $\left\{1,2,\ldots,n-1\right\}$ such that $$\frac{1}{\eta}= \sum\limits_{s=1}^{m}\alpha_{i_{s}}^{-1}$$ where $m=k-1$ or $k-2$.
	\end{theorem}
	
	Based on the above Theorem \ref{NMDSEMGRS1}, a method for constructing MDS EMGRS codes with parameters $[n+1,k,n-k+2]$ or NMDS EMGRS codes with parameters $[n+1,k,n-k+1]$ can be divided into the following four steps:
	
	{\bf Step 1.} Given a positive integer $n\leq q$ and select an $(n-1)$-element subset $$H=\left\{\alpha_{1},\ldots,\alpha_{n-1}\right\}\subseteq\mathbb{F}_{q}^{*};$$ 
	
	{\bf Step 2.} Given a positive integer $2\leq k\leq n$ and calculate the following two sets
	$$\widetilde{H}_{2}=\left\{\sum\limits_{\alpha_{i}\in \mathcal{H}}\alpha_{i}^{-1}:  \mathcal{H}\subseteq H, \#\mathcal{H}=k-2\right\}$$
	and
	$$\widetilde{H}_{3}=\left\{\sum\limits_{\alpha_{i}\in \mathcal{H}}\alpha_{i}^{-1}:  \mathcal{H}\subseteq H, \#\mathcal{H}=k-1\right\};$$
	
	{\bf Step 3.} If $\widetilde{H}_{2}\cup\widetilde{H}_{3}=\mathbb{F}_{q}$, it means $\eta^{-1}\in\widetilde{H}_{2}$ or $\widetilde{H}_{3}$, i.e., there exists a $(k-2)$-element subset or a $(k-1)$-element subset $\mathcal{L}\subseteq\left\{\alpha_{1},\ldots,\alpha_{n-1}\right\}$ such that $\eta^{-1}=\sum\limits_{\alpha_{i}\in \mathcal{L}}\alpha_{i}^{-1}$, then $\mathrm{EMGRS}_{n,k}\left(H,\boldsymbol{v},\eta,1\right)$ is NMDS;
	
	{\bf Step 4.} Otherwise, choice $\eta^{-1}\in\mathbb{F}_{q}\backslash \widetilde{H}_{4}$ with $\widetilde{H}_{4}=\widetilde{H}_{2}\cup\widetilde{H}_{3}$, then $\mathrm{EMGRS}_{n,k}\left(H,\boldsymbol{v},\eta,1\right)$ is MDS.
	\subsection{The weight distributions for a class of NMDS EMGRS codes}\label{sec4.3}
	In this subsection, by the counting formula of the subset sum over $\mathbb{Z}_{n}$, for a special class of NMDS MGRS codes given in Section \ref{sec4.2}, we completely determine the weight distributions.
	\begin{theorem}\label{EMGRSk-1weight}
		If $\mathbb{F}_{q}^{*}=\langle\omega\rangle$ and $\eta=\omega^{t}$ with $1\leq t\leq q-1$, then for the NMDS MGRS code $\mathrm{EMGRS}_{n,k}\left(\mathbb{F}_{q}^{*},\boldsymbol{v},\eta,k-1\right)$, the total
		number $A_{n+1-k}$ of codewords with Hamming weight $n-k$ is 
		$$A_{n+1-k}=\sum_{r\mid \gcd(q-1,k-1)}(-1)^{k-1+\frac{k-1}{r}}\binom{\frac{q-1}{r}}{\frac{k-1}{r}}\sum_{\substack{d\mid (r,t-\frac{k(q-1)}{2})}} \mu\left(\frac{r}{d}\right)d.$$
	\end{theorem}
	{\bf Proof.} By Definition \ref{MGRSdefinition}, it is easy to get  
	$$
	\begin{aligned}
		&\mathrm{EMGRS}_{n+1,k}\left(\boldsymbol{\alpha},\boldsymbol{v},\eta,k-1\right)\\
		\triangleq&\left\{\left(v_{1} f\left(\alpha_{1}\right), \ldots, v_{n-1} f\left(\alpha_{n-1}\right),v_{n}\left(f\left(0\right)+\eta f_{k-1}\right),v_{n+1}f_{k-1}\right) | f(x) \in \mathbb{F}_{q}^{k}[x]\right\}.
	\end{aligned}
	$$
	
	Next, we calculate the number of codewords with minimum
	weight $n+1-k$. For the convenience, we set $\mathcal{I}=\left\{\alpha_{1},\ldots,\alpha_{n-1}\right\}$. Then $\wt\left(\boldsymbol{c}\right)=n+1-k$ if and only if 
	$$\#\left\{\alpha_{i}\in\mathcal{I}\ |\ f(\alpha_{i})=0\right\}=k,\ f\left(0\right)+\eta f_{k-1}\neq 0\ \text{and}\ f_{k-1}\neq 0,$$
	or
	$$\#\left\{\alpha_{i}\in\mathcal{I}\ |\ f(\alpha_{i})=0\right\}=k-1,\ f\left(0\right)+\eta f_{k-1}\neq 0\ \text{and}\ f_{k-1}=0,$$
	or
	$$\#\left\{\alpha_{i}\in\mathcal{I}\ |\ f(\alpha_{i})=0\right\}=k-1,\ f\left(0\right)+\eta f_{k-1}=0\ \text{and}\ f_{k-1}\neq 0,$$
	or
	$$\#\left\{\alpha_{i}\in\mathcal{I}\ |\ f(\alpha_{i})=0\right\}=k-2,\ f\left(0\right)+\eta f_{k-1}=0\ \text{and}\ f_{k-1}=0.$$
	
	{\bf Case 1.} If $f\left(0\right)+\eta f_{k-1}\neq 0$ and $f_{k-1}\neq 0$, then 
	$$
	\begin{aligned}
		&A_{n+1-k}\\
		=&\#\left\{\mathbb{F}_{q}^{k}[x]\cap\left\{f(x)=\lambda\prod\limits_{\alpha_{i}\in\mathcal{J}}\left(x-\alpha_{i}\right)|\ \#\mathcal{J}=k,\mathcal{J}\subseteq\mathcal{I},\lambda\in\mathbb{F}_{q}^{*},f\left(0\right)+\eta f_{k-1}\neq 0,\ f_{k-1}\neq 0\right\}\right\}.
	\end{aligned}
	$$
	Note that for any $f(x)\in \mathbb{F}_{q}^{k}[x]$, $\deg\left(f(x)\right)\leq k-1$, then  $A_{n+1-k}=0.$
	
	{\bf Case 2.} If $f\left(0\right)+\eta f_{k-1}\neq 0$ and $f_{k-1}=0$, then 
	$$
	\begin{aligned}
		&A_{n+1-k}\\
		=&\#\left\{\mathbb{F}_{q}^{k}[x]\cap\left\{f(x)=\lambda\prod\limits_{\alpha_{i}\in\mathcal{J}}\left(x-\alpha_{i}\right)|\ \#\mathcal{J}=k-1,\mathcal{J}\subseteq\mathcal{I},\lambda\in\mathbb{F}_{q}^{*},f\left(0\right)+\eta f_{k-1}\neq 0,\ f_{k-1}= 0\right\}\right\}. 
	\end{aligned}
	$$ 
	In fact,  $f_{k-1}=\lambda=0$ that is contradict with $\lambda\in\mathbb{F}_{q}^{*}$, and so  $A_{n+1-k}=0.$.

	{\bf Case 3.} If $f\left(0\right)+\eta f_{k-1}=0$ and $f_{k-1}\neq 0$, then 
	$$
	\begin{aligned}
		&A_{n+1-k}\\
		=&\#\left\{\mathbb{F}_{q}^{k}[x]\cap\left\{f(x)=\lambda\prod\limits_{\alpha_{i}\in\mathcal{J}}\left(x-\alpha_{i}\right)|\ \#\mathcal{J}=k-1,\mathcal{J}\subseteq\mathcal{I},\lambda\in\mathbb{F}_{q}^{*},f\left(0\right)+\eta f_{k-1}=0,\ f_{k-1}\neq 0\right\}\right\}.
	\end{aligned}
	$$ 
	Note that
	$$f(x)=\lambda\prod\limits_{\alpha_{i}\in\mathcal{J}}\left(x-\alpha_{i}\right)=\lambda\left(x^{k-1}-\sum\limits_{\alpha_{i}\in\mathcal{J}}\alpha_{i}x^{k-2}+\cdots+(-1)^{k-1}\prod\limits_{\alpha_{i}\in\mathcal{J}}\alpha_{i}\right),$$
	thus, $f_{k-1}=\lambda$ and $f(0)=(-1)^{k-1}\cdot\lambda\cdot\prod\limits_{\alpha_{i}\in\mathcal{J}}\alpha_{i}$. Furthermore, we have
	$$\begin{aligned}
		A_{n+1-k}=&\# \left\{\left(\lambda,\mathcal{J}\right)\ |\ (-1)^{k}\prod\limits_{\alpha_{i}\in\mathcal{J}}\alpha_{i}=\eta,\lambda\in\mathbb{F}_{q}^{*},\mathcal{J}\subseteq\mathcal{I}\ \text{and}\ \#\mathcal{J}=k-1\right\}\\
		=&(q-1)\# \left\{\left(\lambda,\mathcal{J}\right)\ |\ (-1)^{k}\prod\limits_{\alpha_{i}\in\mathcal{J}}\alpha_{i}=\eta,\mathcal{J}\subseteq\mathcal{I}\ \text{and}\ \#\mathcal{J}=k-1\right\}.
	\end{aligned}$$
	In the similar proof as that for Theorem $\ref{NMDSMGRSweight}$ {Case 1}, we can get
	$$A_{n+1-k}=\sum_{r\mid \gcd(q-1,k-1)}(-1)^{k-1+\frac{k-1}{r}}\binom{\frac{q-1}{r}}{\frac{k-1}{r}}\sum_{\substack{d\mid (r,t-\frac{k(q-1)}{2})}} \mu\left(\frac{r}{d}\right)d.$$
	
	{\bf Case 4.} If $f\left(0\right)+\eta f_{k-1}=0$ and $f_{k-1}=0$, then 
	$$
	\begin{aligned}
		&A_{n+1-k}\\
		=&\#\left\{\mathbb{F}_{q}^{k}[x]\cap\left\{f(x)=\lambda\prod\limits_{\alpha_{i}\in\mathcal{J}}\left(x-\alpha_{i}\right)|\ \#\mathcal{J}=k-2,\mathcal{J}\subseteq\mathcal{I},\lambda\in\mathbb{F}_{q}^{*},f\left(0\right)+\eta f_{k-1}=0,\ f_{k-1}= 0\right\}\right\}.
	\end{aligned}
	$$ 
	Now, by $f\left(0\right)+\eta f_{k-1}=0$ and $ f_{k-1}= 0$, we have $f(0)=0$, which is contradict with $$f(0)=(-1)^{k-2}\cdot\lambda\cdot\prod\limits_{\alpha_{i}\in\mathcal{J}}\alpha_{i}\in\mathbb{F}_{q}^{*},$$ and so  $A_{n+1-k}=0$.
	
	In conclusions, we have $$A_{n+1-k}=\sum_{r\mid \gcd(q-1,k-1)}(-1)^{k-1+\frac{k-1}{r}}\binom{\frac{q-1}{r}}{\frac{k-1}{r}}\sum_{\substack{d\mid (r,t-\frac{k(q-1)}{2})}} \mu\left(\frac{r}{d}\right)d.$$
	
	This completes the proof of Theorem \ref{EMGRSk-1weight}.
	
	$\hfill\Box$
	\section{The Euclidean  LCD or one-dimensional  Euclidean Hull MGRS codes}\label{sec5}
	Throughout this section, we fix $\mathrm{MGRS}_{k+1,k}\left(\boldsymbol{\alpha},\boldsymbol{v},\eta,t\right)=\mathcal{C}$, $k\mid q-1$, $\mathbb{F}_{q}^{*}=\mathbb{F}_{q}\backslash\left\{0\right\}=\langle\gamma\rangle$ and $\alpha_{i}=\gamma^{\frac{q-1}{k}i}(1\leq i\leq k)$. In this section, by taking four special classes of the vector $\boldsymbol{\alpha}=\left(\alpha_{1},\ldots,\alpha_{n}\right)$, we correspondingly construct four classes of GRL codes  with the parameters $\left[n+\ell,k\right]$, which is either  Euclidean LCD or 1-dimension hull.
	
	\subsection{Main results} 
	\begin{theorem}\label{EHullMGRS1}
		Let $\boldsymbol{\alpha}=\left(\gamma^{\delta}\alpha_{1},\ldots,\gamma^{\delta}\alpha_{k}\right)$ with $1\leq \delta\leq q-1$ and $k\mid q-1$. Then
		$$
		\dim\left(\ehull\left(\mathcal{C}\right)\right)=\begin{cases}
			1,&\ \text{if}\ k=2t,p\nmid k+1,\gamma^{\delta k}k+\eta^2\in\mathbb{F}_{q}^{*} \ \text{and}\ \gamma^{\delta k}k+\frac{k\eta^2}{k+1}=0;\\
			&\ \ \ \ \text{or}\ k\neq 2t \ \text{and}\ p\nmid k+1;\\
			0,&\ \text{if}\ k=2t\ \text{and}\ p\mid k+1;\\
			&\ \ \ \  \text{or}\ k=2t\ \text{and}\ \ \gamma^{\delta k}k+\eta^2=0;\\
			&\ \ \ \ \text{or}\ k=2t,p\nmid k+1,\gamma^{\delta k}k+\eta^2\in\mathbb{F}_{q}^{*} \ \text{and}\ \gamma^{\delta k}k+\frac{k\eta^2}{k+1}\in\mathbb{F}_{q}^{*};\\
			&\ \ \ \ \text{or}\ k\neq 2t \ \text{and}\ p\mid k+1. 
		\end{cases} 
		$$ 
	\end{theorem} 
	\begin{theorem}\label{EHullMGRS2}
		Let $\boldsymbol{\alpha}=\left(0,\gamma^{\delta}\alpha_{1},\ldots,\gamma^{\delta}\alpha_{k}\right)$ with $1\leq \delta\leq q-1$ and $k\mid q-1$. Then
		$$
		\dim\left(\ehull\left(\mathcal{C}\right)\right)=\begin{cases}
			1,&\ \text{if}\ k=2t,p\nmid k+2,\gamma^{\delta k}k+\eta^2\in\mathbb{F}_{q}^{*} \ \text{and}\ \gamma^{\delta k}k+\frac{k\eta^2}{k+2}=0;\\
			&\ \ \ \ \text{or}\ k\neq 2t \ \text{and}\ p\mid k+2.\\
			0,&\ \text{if}\ k=2t\ \text{and}\ p\mid k+2;\\
			&\ \ \ \  \text{or}\ k=2t\ \text{and}\ \ \gamma^{\delta k}k+\eta^2=0;\\
			&\ \ \ \ \text{or}\ k=2t,p\nmid k+2,\gamma^{\delta k}k+\eta^2\in\mathbb{F}_{q}^{*} \ \text{and}\ \gamma^{\delta k}k+\frac{k\eta^2}{k+2}\in\mathbb{F}_{q}^{*};\\
			&\ \ \ \ \text{or}\ k\neq 2t \ \text{and}\ p\mid k+2. 
		\end{cases} 
		$$
	\end{theorem} 
	
	\begin{theorem}\label{EHullMGRS3}
		Let $\boldsymbol{\alpha}=\left(\gamma^{s}\alpha_{1},\ldots,\gamma^{s}\alpha_{k},\gamma^{t}\alpha_{1},\ldots,\gamma^{t}\alpha_{k}\right)$ and $\Delta_{2}=\left(\gamma^{sk}+\gamma^{tk}\right)k$. If $1\leq s\neq t\leq q-1$, $k\mid q-1$, $\frac{q-1}{k}\nmid s-t$ and $v_{2}\left(s-t\right)\neq v_{2}(q-1)-v_{2}(k)-1,$ then
		$$
		\dim\left(\ehull\left(\mathcal{C}\right)\right)\\
		=\begin{cases}
			1,&\ \text{if}\ k=2t,p\nmid 2k+1,\Delta_{2}+\eta^2\in\mathbb{F}_{q}^{*} \ \text{and}\ \Delta_{2}+\frac{2k\eta^2}{2k+1}=0;\\
			&\ \ \ \ \text{or}\ k\neq 2t \ \text{and}\ p\nmid 2k+1;\\
			0,&\ \text{if}\ k=2t\ \text{and}\ p\mid 2k+1;\\
			&\ \ \ \  \text{or}\ k=2t\ \text{and}\ \ \Delta_{2}+\eta^2=0;\\
			&\ \ \ \ \text{or}\ k=2t,p\nmid 2k+1,\Delta_{2}+\eta^2\in\mathbb{F}_{q}^{*} \ \text{and}\ \Delta_{2}+\frac{2k\eta^2}{2k+1}\in\mathbb{F}_{q}^{*};\\ 
			&\ \ \ \ \text{or}\ k\neq 2t \ \text{and}\ p\mid 2k+1. 
		\end{cases}
		$$ 
	\end{theorem} 
	\begin{theorem}\label{EHullMGRS4}
		Let $\gcd(3k,q)=1$ and $$\boldsymbol{\alpha}=\left(\gamma^{\mu}\alpha_{1},\ldots,\gamma^{\mu}\alpha_{k},\gamma^{\mu+1}\alpha_{1},\ldots,\gamma^{\mu+1}\alpha_{k},\gamma^{\mu+2}\alpha_{1},\ldots,\gamma^{\mu+2}\alpha_{k}\right),$$ where $1\leq \mu\leq q-1$, and $\Delta_{\mu}=\left(\gamma^{\mu k}+\gamma^{(\mu+1)k}+\gamma^{(\mu+2)k}\right)k$. If $1\leq s\neq t\leq q-1$, $k\mid q-1$ and $q-1\notin\left\{k,2k,3k\right\}$, then $$
		\dim\left(\ehull\left(\mathcal{C}\right)\right)=\begin{cases}
			1,&\ \text{if}\ k=2t, p\nmid 3k+1, \Delta_{\mu }+\eta^2\in\mathbb{F}_{q}^{*} \ \text{and}\ \Delta_{\mu }+\frac{3k\eta^2}{3k+1}=0;\\
			&\ \ \ \ \text{or}\ k\neq 2t \ \text{and}\ p\nmid 3k+1;\\
			0,&\ \text{if}\ k=2t\ \text{and}\ p\mid 3k+1;\\
			&\ \ \ \  \text{or}\ k=2t\ \text{and}\ \ \Delta_{\mu }+\eta^2=0;\\
			&\ \ \ \ \text{or}\ k=2t,p\nmid 3k+1,\Delta_{\mu }+\eta^2\in\mathbb{F}_{q}^{*} \ \text{and}\ \Delta_{\mu }+\frac{3k\eta^2}{3k+1}\in\mathbb{F}_{q}^{*};\\ 
			&\ \ \ \ \text{or}\ k\neq 2t \ \text{and}\ p\mid 3k+1. 
		\end{cases}
		$$ 
	\end{theorem}
	\subsection{The Proofs for Theorems \ref{EHullMGRS1}-\ref{EHullMGRS4}}
	
	In this subsection, we only provide the detailed proof for Theorem \ref{EHullMGRS1}.
	Since Theorems \ref{EHullMGRS2}–\ref{EHullMGRS4} can be finished via the similar method,  and so their proofs are omitted for brevity.
	
	{\bf The proof of Theorem \ref{EHullMGRS1}.} 
	
	By Definition \ref{MGRSdefinition}, we can get $$\boldsymbol{G}_{\boldsymbol{v},k,t} =\boldsymbol{G}_{\boldsymbol{1},k,t} \cdot\diag\left\{v_1,\ldots,v_{n-1},1\right\},$$ and then
	$$\rank\left(\boldsymbol{G}_{\boldsymbol{v},k,t}\boldsymbol{G}^{T_e}_{\boldsymbol{v},k,t}\right) =\rank\left(\boldsymbol{G}_{\boldsymbol{1},k,t}\boldsymbol{G}^{T_e}_{\boldsymbol{1},k,t} \cdot\diag\left\{v_{1}^{2},\ldots,v_{n-1}^{2},1\right\}\right)=\rank\left(\boldsymbol{G}_{\boldsymbol{1},k,t}\boldsymbol{G}^{T_e}_{\boldsymbol{1},k,t}\right).$$
	Now, by Lemma \ref{LCDMDSorNMDSequivalent}, we only calculate $\rank\left(\boldsymbol{G}_{\boldsymbol{1},k,t}\boldsymbol{G}^{T_e}_{\boldsymbol{1},k,t}\right)$. Note that
	$$
	\footnotesize\begin{aligned}
		&\boldsymbol{G}_{\boldsymbol{1},k,t}\boldsymbol{G}_{\boldsymbol{1},k,t}^{T_{e}}\\
		=&\begin{pmatrix}
			k+1&\sum\limits_{i=1}^{k}\gamma^{\delta}\alpha_{i}&\cdots&\sum\limits_{i=1}^{k}\left(\gamma^{\delta}\alpha_{i}\right)^{t-1}&\color{red}{\sum\limits_{i=1}^{k}\left(\gamma^{\delta}\alpha_{i}\right)^{t}+\eta}&\sum\limits_{i=1}^{k}\left(\gamma^{\delta}\alpha_{i}\right)^{t+1}&\cdots&\sum\limits_{i=1}^{k}\left(\gamma^{\delta}\alpha_{i}\right)^{k-1}\\
			\sum\limits_{i=1}^{k}\left(\gamma^{\delta}\alpha_{i}\right)&\sum\limits_{i=1}^{k}\left(\gamma^{\delta}\alpha_{i}\right)^{2}&\cdots&\sum\limits_{i=1}^{k}\left(\gamma^{\delta}\alpha_{i}\right)^{t}&\sum\limits_{i=1}^{k}\left(\gamma^{\delta}\alpha_{i}\right)^{t+1}&\sum\limits_{i=1}^{k}\left(\gamma^{\delta}\alpha_{i}\right)^{t+2}&\cdots&\sum\limits_{i=1}^{k}\left(\gamma^{\delta}\alpha_{i}\right)^{k}\\
			\vdots&\vdots& &\vdots&\vdots\\
			\sum\limits_{i=1}^{k}\left(\gamma^{\delta}\alpha_{i}\right)^{t-1}&\sum\limits_{i=1}^{k}\left(\gamma^{\delta}\alpha_{i}\right)^{t}&\cdots&\sum\limits_{i=1}^{k}\left(\gamma^{\delta}\alpha_{i}\right)^{2t-2}&\sum\limits_{i=1}^{k}\left(\gamma^{\delta}\alpha_{i}\right)^{2t-1}&\sum\limits_{i=1}^{k}\left(\gamma^{\delta}\alpha_{i}\right)^{2t}&\cdots&\sum\limits_{i=1}^{k}\left(\gamma^{\delta}\alpha_{i}\right)^{k+t-2}\\
			\color{red}{\sum\limits_{i=1}^{k}\left(\gamma^{\delta}\alpha_{i}\right)^{t}+\eta}&\sum\limits_{i=1}^{k}\left(\gamma^{\delta}\alpha_{i}\right)^{t+1}&\cdots&\sum\limits_{i=1}^{k}\left(\gamma^{\delta}\alpha_{i}\right)^{2t-1}&\color{red}{\sum\limits_{i=1}^{k}\left(\gamma^{\delta}\alpha_{i}\right)^{2t}+\eta^{2}}&\sum\limits_{i=1}^{k}\left(\gamma^{\delta}\alpha_{i}\right)^{2t+1}&\cdots&\sum\limits_{i=1}^{k}\left(\gamma^{\delta}\alpha_{i}\right)^{k+t-1}\\ 
			\sum\limits_{i=1}^{k}\left(\gamma^{\delta}\alpha_{i}\right)^{t+1}&\sum\limits_{i=1}^{k}\left(\gamma^{\delta}\alpha_{i}\right)^{t+2}&\cdots&\sum\limits_{i=1}^{k}\left(\gamma^{\delta}\alpha_{i}\right)^{2t}&\sum\limits_{i=1}^{k}\left(\gamma^{\delta}\alpha_{i}\right)^{2t+1}&\sum\limits_{i=1}^{k}\left(\gamma^{\delta}\alpha_{i}\right)^{2t+2}&\cdots&\sum\limits_{i=1}^{k}\left(\gamma^{\delta}\alpha_{i}\right)^{k+t}\\
			\vdots&\vdots& &\vdots&\vdots\\
			\sum\limits_{i=1}^{k}\left(\gamma^{\delta}\alpha_{i}\right)^{k-1}&\sum\limits_{i=1}^{k}\left(\gamma^{\delta}\alpha_{i}\right)^{k}&\cdots&\sum\limits_{i=1}^{k}\left(\gamma^{\delta}\alpha_{i}\right)^{k+t-2}&\sum\limits_{i=1}^{k}\left(\gamma^{\delta}\alpha_{i}\right)^{k+t-1}&\sum\limits_{i=1}^{k}\left(\gamma^{\delta}\alpha_{i}\right)^{k+t}&\cdots&\sum\limits_{i=1}^{k}\left(\gamma^{\delta}\alpha_{i}\right)^{2k-2}
		\end{pmatrix}_{k\times k}.
	\end{aligned}
	$$
	
	Now, based on $k=2t$ or not, we have the following two cases.
	
	{\bf Case 1.} If $2t=k$, then by Lemma \ref{Fq*sum}, we have 
	$$
	\boldsymbol{G}_{\boldsymbol{1},k,t}\boldsymbol{G}_{\boldsymbol{1},k,t}^{T_{e}}=\begin{pmatrix}
		k+1&0&\cdots&0&\color{red}{\eta}&0&\cdots&0&0\\
		0&0&\cdots&0&0&0&\cdots&0&\gamma^{\delta k}k\\
		\vdots&\vdots& &\vdots&\vdots&\vdots& &\vdots&\vdots\\
		
		0&0&\cdots&0&0&\gamma^{\delta k}k&\cdots&0&0\\
		
		\color{red}{\eta}&0&\cdots&0&\color{red}{\gamma^{\delta k}k+\eta^{2}}&0&\cdots&0&0\\ 
		0&0&\cdots&\gamma^{\delta k}k&0&0&\cdots&0&0\\
		\vdots&\vdots& &\vdots&\vdots&\vdots& &\vdots\\
		0&\gamma^{\delta k}k&\cdots&0&0&0&\cdots&0&0
	\end{pmatrix}_{k\times k},
	$$
	it's easy to get 
	$$\rank\left(\boldsymbol{G}_{\boldsymbol{1},k,t}\boldsymbol{G}_{\boldsymbol{1},k,t}^{T_{e}}\right)=\begin{cases}
		k-1,&\ \text{if}\ k=2t,p\nmid k+1,\gamma^{\delta k}k+\eta^2\in\mathbb{F}_{q}^{*} \ \text{and}\ \gamma^{\delta k}k+\frac{k\eta^2}{k+1}=0;\\ 
		k,&\ \text{if}\ k=2t\ \text{and}\ p\mid k+1;\\
		&\ \ \ \  \text{or}\ k=2t\ \text{and}\ \ \gamma^{\delta k}k+\eta^2=0;\\
		&\ \ \ \ \text{or}\ k=2t,p\nmid k+1,\gamma^{\delta k}k+\eta^2\in\mathbb{F}_{q}^{*} \ \text{and}\ \gamma^{\delta k}k+\frac{k\eta^2}{k+1}\in\mathbb{F}_{q}^{*}. 
	\end{cases}$$ 
	Furthermore, by Lemma \ref{EHhull}, we have
	$$
	\dim\left(\ehull\left(\mathcal{C}\right)\right)=\begin{cases}
		1,&\ \text{if}\ k=2t, p\nmid 3k+1, \Delta_{\mu }+\eta^2\in\mathbb{F}_{q}^{*} \ \text{and}\ \Delta_{\mu }+\frac{3k\eta^2}{3k+1}=0;\\
		&\ \ \ \ \text{or}\ k\neq 2t \ \text{and}\ p\nmid 3k+1;\\
		0,&\ \text{if}\ k=2t\ \text{and}\ p\mid 3k+1;\\
		&\ \ \ \  \text{or}\ k=2t\ \text{and}\ \ \Delta_{\mu }+\eta^2=0;\\
		&\ \ \ \ \text{or}\ k=2t,p\nmid 3k+1,\Delta_{\mu }+\eta^2\in\mathbb{F}_{q}^{*} \ \text{and}\ \Delta_{\mu }+\frac{3k\eta^2}{3k+1}\in\mathbb{F}_{q}^{*};\\ 
		&\ \ \ \ \text{or}\ k\neq 2t \ \text{and}\ p\mid 3k+1, 
	\end{cases}
	$$
	
	{\bf Case 2.} If $2t\neq k$, without loss of generality, we set $2t<k$, then by Lemma \ref{Fq*sum}, we have 
	$$
	\boldsymbol{G}_{\boldsymbol{1},k,t}\boldsymbol{G}_{\boldsymbol{1},k,t}^{T_{e}}=\begin{pmatrix}
		k+1&0&\cdots&0&\color{red}{\eta}&0&\cdots&0&0&0&\cdots&0&0\\
		0&0&\cdots&0&0&0&\cdots&0&0&0&\cdots&0&\gamma^{\delta k}k\\
		\vdots&\vdots& &\vdots&\vdots&\vdots& &\vdots&\vdots&\vdots& &\vdots&\vdots\\
		
		0&0&\cdots&0&0&0&\cdots&0&0&\gamma^{\delta k}k&\cdots&0&0\\
		
		\color{red}{\eta}&0&\cdots&0&\color{red}{\eta^{2}}&0&\cdots&0&\gamma^{\delta k}k&0&\cdots&0&0\\  
		
		0&0&\cdots&0&0&0&\cdots&\gamma^{\delta k}k&0&0&\cdots&0&0\\  
		\vdots&\vdots& &\vdots&\vdots&\vdots& &\vdots&\vdots&\vdots& &\vdots&\vdots\\
		0&0&\cdots&0&0&\gamma^{\delta k}k&\cdots&0&0&0&\cdots&0&0\\ 
		0&0&\cdots&0&\gamma^{\delta k}k&0&\cdots&0&0&0&\cdots&0&0\\ 
		0&0&\cdots&\gamma^{\delta k}k&0&0&\cdots&0&0&0&\cdots&0&0\\   
		\vdots&\vdots& &\vdots&\vdots&\vdots& &\vdots&\vdots&\vdots& &\vdots&\vdots\\
		0&\gamma^{\delta k}k&\cdots&0&0&0&\cdots&0&0&0&\cdots&0&0
	\end{pmatrix}_{k\times k},
	$$
	it's easy to get
	$$\rank\left(\boldsymbol{G}_{\boldsymbol{1},k,t}\boldsymbol{G}_{\boldsymbol{1},k,t}^{T_{e}}\right)=\begin{cases}
		k-1,&\ \text{if}\ k\neq 2t,p\nmid k+1;\\
		k,&\ \text{if}\ k\neq 2t,p\mid k+1.
	\end{cases}$$
	Furthermore, by Lemma \ref{EHhull}, we have
	$$\dim\left(\ehull\left(\mathrm{MGRS}_{k+1,k}\left(\boldsymbol{\alpha},\boldsymbol{v},\eta,t\right)\right)\right)=\begin{cases}
		1,&\ \text{if}\ k\neq 2t,p\nmid k+1;\\
		0,&\ \text{if}\ k\neq 2t,p\mid k+1.
	\end{cases}$$
	
	This completes the proof of Theorem $\ref{EHullMGRS1}$.
	
	\section{A class of MGRS codes with  flexible Hermitian hull dimensions}\label{sec6}
	In this section, by introducing a variable $\delta$ and a positive integer $m$, we construct a class of MGRS codes with length $m(q+1)$. These codes possess a $(\delta-2)$-dimension Hermitian hull when $\delta\geq 2$, and are Hermitian LCD or have a one-dimensional Hermitian hull when $\delta=1$, i.e., we constructively prove that there exists an MGRS code with flexible length and Hermitian hull dimension.
	
	\subsection{Three essential results about the symmetric polynomial}
	
	In this subsection, we present three key propositions about the elementary symmetric polynomial and the complete symmetric polynomial. 
	
	Firstly, we give a formula of the elementary symmetric polynomial with the set $\left\{1,x,x^{2},\ldots,x^{n-1}\right\}$ and $1-x^{j}\neq 0(i=1,\ldots,n)$.
	\begin{proposition}\label{ESPAGS}
		For any indeterminate $x$ with $1-x^{j}\neq 0(i=1,\ldots,n)$, let $I=\left\{1,x,x^{2},\ldots,x^{n-1}\right\}$, then
		$$\sigma_{i}(I)=x^{\frac{i(i-1)}{2}}\prod\limits_{j=1}^{i}\frac{1-x^{n-i+j}}{1-x^{j}}(i=1,\ldots,n).$$
	\end{proposition}
	{\bf Proof.} Note that
	\begin{equation}\label{ESPE}
		\prod\limits_{i=0}^{n-1}\left(1+x^{i}z\right)=\sum\limits_{i=0}^{n}\sigma_{i}(I)z^{i}.
	\end{equation}
	And by the $q$-binomial theorem in \cite[Page 117, Exercise 7]{FengRQCombinatorics}, we can get 
	\begin{equation}\label{qtheorem}
		\prod\limits_{i=0}^{n-1}\left(1+x^{i}z\right)=\sum\limits_{i=0}^{n}x^{\binom{i}{2}}\begin{bmatrix}n\\i\end{bmatrix}_{x}z^{i}.
	\end{equation} 
	Now, by comparing the coefficients of $x^{i}$ on both sides of the equations \eqref{ESPE} and \eqref{qtheorem}, we have
	$$\sigma_{i}(I)=x^{\binom{i}{2}}\begin{bmatrix}n\\i\end{bmatrix}_{x}.$$
	Thus, by the $x$-analog of the binomial coefficient given in \cite[Page 87, Definition 4.1.8]{FengRQCombinatorics},
	$$\begin{aligned}
		x^{\binom{i}{2}}\begin{bmatrix}n\\i\end{bmatrix}_{x}=&x^{\frac{i(i-1)}{2}}\frac{[n]_{x}!}{[i]_{x}![n-i]_{x}!}=x^{\frac{i(i-1)}{2}}\frac{\prod\limits_{j=1}^{n}[j]_{x}}{\prod\limits_{j=1}^{i}[j]_{x}\prod\limits_{j=1}^{n-i}[j]_{x}}=x^{\frac{i(i-1)}{2}}\frac{\prod\limits_{j=1}^{n}\frac{1-x^{j}}{1-x}}{\prod\limits_{j=1}^{i}\frac{1-x^{j}}{1-x}\prod\limits_{j=1}^{n-i}\frac{1-x^{j}}{1-x}}\\
		=&x^{\frac{i(i-1)}{2}}\frac{\prod\limits_{j=n-i+1}^{n}\left(1-x^{j}\right)}{\prod\limits_{j=1}^{i}\left(1-x^{j}\right)}=x^{\frac{i(i-1)}{2}}\frac{\prod\limits_{j=1}^{i}\left(1-x^{n-i+j}\right)}{\prod\limits_{j=1}^{i}\left(1-x^{j}\right)}=x^{\frac{i(i-1)}{2}}\prod\limits_{j=1}^{i}\frac{1-x^{n-i+j}}{1-x^{j}},
	\end{aligned}$$
	i.e.,
	$$\sigma_{i}(I)=x^{\frac{i(i-1)}{2}}\prod\limits_{j=1}^{i}\frac{1-x^{n-i+j}}{1-x^{j}}.$$
	
	This completes the proof of Proposition \ref{ESPAGS}.
	
	$\hfill\Box$
	
	The following Proposition \ref{ESPCSP} shows that the non-vanishing of the elementary symmetric polynomial of a subset $I_{m}$ of $\mathbb{F}_{q}^{*}$ is equivalent to the non-vanishing of the complete symmetric polynomial of its complement set.
	\begin{proposition}\label{ESPCSP}
		Let $I_{m}=\left\{c_{1},c_{2},\ldots,c_{m}\right\}$ and $J_{m}=\mathbb{F}_{q}^{*}\backslash I_{m}.$ Then for any $1\leq i\leq m$, 
		$$S_{i}\left(I_{m}\right)\neq 0\Longleftrightarrow \sigma_{i}\left(J_{m}\right)\neq 0.$$
	\end{proposition} 
	{\bf Proof.} Firstly, for any $x\in\mathbb{F}_{q}^{*}=I_{m}\cup J_{m}$, $x^{q-1}=1$, and $I_{m}\cap J_{m}=\emptyset$, we have \begin{equation}\label{yq-1-1}
		\prod _{x_i \in I_m}(y-x_i)\prod _{x_j \in J_m} (y-x_j) = y ^{q-1} - 1.
	\end{equation}
	Note that
	$$\prod _{x_i \in I_m}(y-x_i)=y^{q-1-m}-\sigma_{1}\left(I_{m}\right)y^{q-2-m}+\sigma_{2}\left(I_{m}\right)y^{q-3-m}-\ldots+(-1)^{q-1-m}\sigma_{q-1-m}\left(I_{m}\right)$$
	and
	$$\prod _{x_j \in J_m}(y-x_j)=y^{q-1-m}-\sigma_{1}\left(J_{m}\right)y^{q-2-m}+\sigma_{2}\left(J_{m}\right)y^{q-3-m}-\ldots+(-1)^{q-1-m}\sigma_{q-1-m}\left(J_{m}\right).$$
	By comparing the coefficients of $y^{t}(1\leq t \leq q-2)$ on both sides for the equation \eqref{yq-1-1}, we have
	$$\sum\limits_{i=0}^{t}\sigma_{i}\left(I_{m}\right)\sigma_{t-i}\left(J_{m}\right)=0.$$
	
	Secondly, for any non-empty set $I_{m}$ and positive integer $N$, by Remark \ref{DofESpolynomial}, we have
	$$\sum\limits_{i=0}^{N}(-1)^{i}\sigma_{i}\left(I_{m}\right)S_{N-i}\left(I_{m}\right)=0.$$
	Then 
	$$S_{N}\left(I_{m}\right)=\sum\limits_{i=1}^{N}(-1)^{i}\sigma_{i}\left(I_{m}\right)S_{N-i}\left(I_{m}\right).$$ 
	
	Next, by using the mathematical induction, we prove that for any $1\leq i\leq m$,
	$$S_{i}\left(I_{m}\right)\neq 0\Longleftrightarrow \sigma_{i}\left(J_{m}\right)\neq 0.$$
	
	For $j=1$, by $\sum\limits_{x\in\mathbb{F}_{q}^{*}}x=0$, we have $\sum\limits_{x\in I_{m}}x+\sum\limits_{x\in J_{m}}x=0$, i.e., $\sigma_{1}\left(I_{m}\right)+\sigma_{1}\left(J_{m}\right)=0$, it means 
	$$\sigma_{1}\left(I_{m}\right)\neq 0\Longleftrightarrow -\sigma_{1}\left(J_{m}\right)\neq 0,$$
	furthermore,
	$$S_{1}\left(I_{m}\right)\neq 0\Longleftrightarrow \sigma_{1}\left(J_{m}\right)\neq 0.$$
	
	Now, we assume that for any $j\leq i-1$ with $1\leq i\leq m$, we have $$S_{j}\left(I_{m}\right)=(-1)^{j}\sigma_{j}\left(J_{m}\right).$$
	Then, for $j=i$, we have
	$$
	\begin{aligned}
		S_{i}\left(I_{m}\right)=&\sum\limits_{\ell=1}^{i}(-1)^{\ell}\sigma_{\ell}\left(I_{m}\right)S_{i-\ell}\left(I_{m}\right)\\
		=&\sum\limits_{\ell=1}^{i}(-1)^{\ell}\sigma_{\ell}\left(I_{m}\right)\cdot(-1)^{i-\ell}\sigma_{i-\ell}\left(J_{m}\right)\\
		=&(-1)^{i}\sum\limits_{\ell=1}^{i}\sigma_{\ell}\left(I_{m}\right)\sigma_{i-\ell}\left(J_{m}\right)\\
		=&(-1)^{i}\left(\sum\limits_{\ell=0}^{i}\sigma_{\ell}\left(I_{m}\right)\sigma_{i-\ell}\left(J_{m}\right)-\sigma_{i}\left(J_{m}\right)\right)\\
		=&(-1)^{i+1}\sigma_{i}\left(J_{m}\right).
	\end{aligned}
	$$
	And so, by the induction hypothesis, for $1\leq i\leq m$, we can get
	$$S_{i}\left(I_{m}\right)\neq 0\Longleftrightarrow \sigma_{i}\left(J_{m}\right)\neq 0.$$
	
	This completes the proof of Proposition \ref{ESPCSP}.
	
	$\hfill\Box$
	
	The following Proposition \ref{ESPneq0} constructively prove that there exists some $m$-subset $J_{m}\subseteq\mathbb{F}_{q}^{*}$ such that $\sigma_{i}\left(J_{m}\right)\neq 0(i=1,2,\ldots,k-\delta)$ over $\mathbb{F}_{q}^{*}$.  
	\begin{proposition}\label{ESPneq0}
		For any given positive integer $\delta$ and $k$ with $1\leq \delta \leq k-1$. Let $q$ be a prime power with $q\geq 2(k-\delta)+1$ and $m$ an positive integer with $k-\delta\leq m\leq q-1-(k-\delta)$. Then there exists some $m$-subset $J_{m}\subseteq\mathbb{F}_{q}^{*}$ such that
		$$\sigma_{i}\left(J_{m}\right)\neq 0(i=1,2,\ldots,k-\delta).$$
	\end{proposition} 
	{\bf Proof.} If $m=k-\delta$, for the $(k-\delta)$-subset $J_{k-\delta}=\left\{\omega,\omega^{2},\ldots,\omega^{k-\delta}\right\}\subseteq\mathbb{F}_{q}^{*}=\langle\omega\rangle$, by Lemma \ref{ESPAGS}, we have
	$$\sigma_{i}\left(J_{k-\delta}\right)=\omega^{\frac{i(i-1)}{2}}\prod\limits_{j=1}^{i}\frac{1-\omega^{k-\delta-i+j}}{1-\omega^{j}}(i=1,2,\ldots,k-\delta).$$
	Note that 
	$$1\leq j\leq i\leq k-\delta\leq q-1-(k-\delta)<q-1,$$
	$$ k-\delta-i+j\leq k-\delta<q-1,$$ 
	and
	$$k-\delta-i+j\geq k-\delta-i+1\geq 1,$$
	then
	$$\sigma_{i}\left(J_{k-\delta}\right)=x^{\frac{i(i-1)}{2}}\prod\limits_{j=1}^{i}\frac{1-\omega^{k-\delta-i+j}}{1-\omega^{j}}\neq 0(i=1,2,\ldots,k-\delta),$$
	i.e., there exists a $(k-\delta)$-subset $J_{k-\delta}\subseteq\mathbb{F}_{q}^{*}$ such that
	$$\sigma_{i}\left(J_{k-\delta}\right)\neq 0(i=1,2,\ldots,k-\delta).$$
	
	Assume that for any given positive integer $t$ with $k-\delta \leq t \leq m-1$, there exists some valid $t$ subset $J_{t}\subseteq\mathbb{F}_{q}^{*}$ such that
	$$\sigma_{i}\left(J_{t}\right)\neq 0(i=1,2,\ldots,k-\delta).$$ 
	Now, based on the above set $J_{t}$, we construct a $(t+1)$-element subset $J_{t+1}$ such that
	$$\sigma_{i}\left(J_{t+1}\right)\neq 0(i=1,2,\ldots,k-\delta).$$ 
	Note that for any $x\in\mathbb{F}_{q}^{*}\backslash J_{t}$,
	$$\sigma_{i}\left(J_{t}\cup\left\{x\right\}\right)=\sigma_{i}\left(J_{t}\right)+x\sigma_{i-1}\left(J_{t}\right)(i=1,2,\ldots,k-\delta).$$ 
	Then for any given $i$, 
	$$\sigma_{i}\left(J_{t}\cup\left\{x\right\}\right)=0\Longleftrightarrow\sigma_{i}\left(J_{t}\right)+x\sigma_{i-1}\left(J_{t}\right)=0\Longleftrightarrow x=-\frac{\sigma_{i}\left(J_{t}\right)}{\sigma_{i-1}\left(J_{t}\right)}.$$
	And by 
	$$
	\begin{aligned}
		|\mathbb{F}_{q}^{*}\backslash J|=& q-1-(t+k-\delta)\\
		\geq& q-1-(m-1+k-\delta)\\
		\geq& q-1-\left(q-1-(k-\delta)-1+k-\delta\right)\\
		=& 1,
	\end{aligned}
	$$
	where $J=J_{t}\cup \bigcup\limits_{i=1}^{k-\delta}\left\{-\frac{\sigma_{i}\left(J_{t}\right)}{\sigma_{i-1}\left(J_{t}\right)}\right\}$, there exists at least one element $j\in\mathbb{F}_{q}^{*}$ such that 
	$$\sigma_{i}\left(J_{t}\cup\left\{j\right\}\right)\neq 0(i=1,2,\ldots,k-\delta).$$
	
	This completes the proof of Proposition \ref{ESPneq0}.
	
	$\hfill\Box$
	
	\begin{remark}
		By taking $k-\delta=2$ in Proposition \ref{ESPneq0}, the corresponding result is just Lemma IV.14 in \cite{ZhouHYGRL2026}.
	\end{remark}
	\iffalse By Lemma 
	\begin{theorem}
		For given positive integer $\delta$ with $1\leq \delta \leq k-1$. Let $q$ be a prime power with $q\geq 2(k-\delta)+1$ and $m$ an positive integer with $k-\delta\leq m\leq q-1-(k-\delta)$. Then there exist a $m$-subset $I_{m}\subseteq\mathbb{F}_{q}^{*}$ such that
		$$S_{i}\left(I_{m}\right)\neq 0(i=1,2,\ldots,k-\delta).$$
	\end{theorem}\fi
	
	\subsection{Main results}
	In this subsection, basing on Propositions \ref{ESPAGS}-\ref{ESPneq0}, we constructively prove that there exist MGRS codes with flexible length and Hermitian hull dimensions as the following
	\begin{theorem}\label{MGRSHhull}
		For any given positive integer $k$, $m$ and $\delta$ with $1\leq\delta\leq \min\left\{k-1,m\right\}$. If $q\geq 2k-2\delta+1$, $k-\delta\leq m\leq q-1-(k-\delta)$, then the following two statements are true.
		
		$(1)$ If $\delta=1$, then there exists an $[m(q+1)+1,k]_{q^2}$ MGRS code with $$
		\dim\left(\hhull\left(\mathrm{MGRS}_{m(q+1)+1,k}\left(\boldsymbol{\alpha},\boldsymbol{v},\eta,t\right)\right)\right)=0\ \text{or}\ 1;
		$$
		
		$(2)$ If $\delta\geq 2$, then there exists an $[m(q+1)+1,k]_{q^2}$ MGRS code with $$
		\dim\left(\hhull\left(\mathrm{MGRS}_{m(q+1)+1,k}\left(\boldsymbol{\alpha},\boldsymbol{v},\eta,t\right)\right)\right)=\delta-2.
		$$  
	\end{theorem} 
	{\bf Proof.} Firstly, for any $[n+1,k]_{q^2}$ MGRS code with the generator matrix \begin{equation}\label{Gvkt}
		\boldsymbol{G}_{\boldsymbol{v},k,t}= \begin{pmatrix}
			v_{1}&v_{2}&\cdots&v_{n}&v_{n+1}\\
			v_{1}\alpha_{1}&v_{2}\alpha_{2}&\cdots&v_{n}\alpha_{n}&0\\
			
			\vdots&\vdots& &\vdots&\vdots\\
			v_{1}\alpha_{1}^{t-1}&v_{2}\alpha_{2}^{t-1}&\cdots&v_{n}\alpha_{n}^{t-1}&0\\
			
			v_{1}\alpha_{1}^{t}&v_{2}\alpha_{2}^{t}&\cdots&v_{n}\alpha_{n}^{t}&v_{n+1}\eta\\
			
			v_{1}\alpha_{1}^{t+1}&v_{2}\alpha_{2}^{t+1}&\cdots&v_{n}\alpha_{n}^{t+1}&0\\
			\vdots&\vdots& &\vdots&\vdots\\ 
			v_{1}\alpha_{1}^{k-1}&v_{2}\alpha_{2}^{k-1}&\cdots&v_{n}\alpha_{n}^{k-1}&0
		\end{pmatrix}_{k\times (n+1)},
	\end{equation} it's easy to know that
	$$
	\begin{aligned}
		&\boldsymbol{G}_{\boldsymbol{v},k,t}\boldsymbol{G}_{\boldsymbol{v},k,t}^{T_{h}}\\
		=&\begin{pmatrix}
			S_{0}^{(H)}+v_{n+1}^{1+q}&S_{q}^{(H)}&\cdots&S_{(t-1)q}^{(H)}&S_{tq}^{(H)}+v_{n+1}^{1+q}\eta^{q}&S_{(t+1)q}^{(H)}&\cdots&S_{(k-1)q}^{(H)}\\
			S_{1}^{(H)}&S_{1+q}^{(H)}&\cdots&S_{1+(t-1)q}^{(H)}&S_{1+tq}^{(H)}&S_{1+(t+1)q}^{(H)}&\cdots&S_{1+(k-1)q}^{(H)}\\
			\vdots&\vdots& &\vdots&\vdots\\
			S_{t-1}^{(H)}&S_{t-1+q}^{(H)}&\cdots&S_{t-1+(t-1)q}^{(H)}&S_{t-1+tq}^{(H)}&S_{t-1+(t+1)q}^{(H)}&\cdots&S_{t-1+(k-1)q}^{(H)}\\
			S_{t}^{(H)}+v_{n+1}^{1+q}\eta&S_{t+q}^{(H)}&\cdots&S_{t+(t-1)q}^{(H)}&S_{t+tq}^{(H)}+v_{n+1}^{1+q}\eta^{1+q}&S_{t+(t+1)q}^{(H)}&\cdots&S_{t+(k-1)q}^{(H)}\\
			S_{t+1}^{(H)}&S_{t+1+q}^{(H)}&\cdots&S_{t+1+(t-1)q}^{(H)}&S_{t+1+tq}^{(H)}&S_{t+1+(t+1)q}^{(H)}&\cdots&S_{t+1+(k-1)q}^{(H)}\\
			\vdots&\vdots& &\vdots&\vdots\\
			S_{k-1}^{(H)}&S_{k-1+q}^{(H)}&\cdots&S_{k-1+(t-1)q}^{(H)}&S_{k-1+tq}^{(H)}&S_{k-1+(t+1)q}^{(H)}&\cdots&S_{k-1+(k-1)q}^{(H)}\\
		\end{pmatrix}
	\end{aligned}
	$$
	where $S_{t}^{(H)}=\sum\limits_{i=1}^{n}v_{i}^{1+q}\alpha_{i}^{t}$.
	
	By Lemma \ref{ESPneq0}, we know that the $m$-subset $I_{m}$ of $\mathbb{F}_{q}^{*}$ satisfying $S_{i}\left(I_{m}\right)\neq 0(i=1,\ldots,k-\delta)$ always exits. And so, we can set $n=m(q+1)$ and the $m$-subset $I_{m}=\left\{c_{1},c_{2},\ldots,c_{m}\right\}\subseteq\mathbb{F}_{q}^{*}$ such that
	$$S_{i}\left(I_{m}\right)\neq 0(i=1,\ldots,k-\delta).$$ 
	
	We define $\beta_{j}\in\mathbb{F}_{q^{2}}^{*}$ with $c_{j}=\beta_{j}^{q+1}\in\mathbb{F}_{q}^{*}$, and  $$\boldsymbol{\alpha}=\left(\alpha_{1},\ldots,\alpha_{m(q+1)}\right)=\bigcup\limits_{j=1}^{m}\beta_{j}\cdot U_{q+1}=\bigcup\limits_{j=1}^{m}C_{j},$$
	where $U_{q+1}$ is the subgroup of $\mathbb{F}_{q^{2}}^{*}$ with order $q+1$. 
	
	By the permutation equivalence of linear codes, without loss of generality, we can assume 
	$$
	\boldsymbol{\alpha }=\bigl( \,\underset{C1}{\underbrace{\alpha _1,...,\alpha _{q+1}}},\,\underset{C_2}{\underbrace{\alpha _{q+2},...,\alpha _{2\left( q+1 \right)}}},\,...,\,\underset{C_m}{\underbrace{\alpha _{\left( m-1 \right) \left( q+1 \right) +1},...,\alpha _{m\left( q+1 \right)}}}\, \bigr) .
	$$
	And for any given $1\leq i\leq m$ and $1\leq \delta \leq \min\left\{k-1,m\right\}$, we set
	$$v_{(i-1)(q+1)+s}^{1+q}=u_{j}c_{j}^{m-\delta }\in\mathbb{F}_{q}^{*}(s=1,2,\ldots,q+1),$$
	where $u_{j}=\prod\limits_{r=1, j \neq r}^{m}\left(c_{j}-c_{r}\right)^{-1}$.
	
	And so
	$$
	\begin{aligned}
		S_{r+sq}^{(H)}=&\sum\limits_{i=1}^{n}v_{i}^{1+q}\alpha_{i}^{r+sq}=\sum\limits_{i=1}^{m(q+1)}v_{i}^{1+q}\alpha_{i}^{r+sq}\\
		=&\sum\limits_{i=1}^{q+1}v_{i}^{1+q}\alpha_{i}^{r+sq}+\sum\limits_{i=q+2}^{2(q+1)}v_{i}^{1+q}\alpha_{i}^{r+sq}+\cdots+\sum\limits_{i=1}^{(m-1)(q+1)+1}v_{i}^{m(q+1)}\alpha_{i}^{r+sq}\\
		=&\sum\limits_{i=1}^{q+1}u_{1}c_{1}^{m-\delta }\alpha_{i}^{r+sq}+\sum\limits_{i=q+2}^{2(q+1)}u_{2}c_{2}^{m-\delta }\alpha_{i}^{r+sq}+\cdots+\sum\limits_{i=1}^{(m-1)(q+1)+1}u_{m}c_{m}^{m-\delta }\alpha_{i}^{r+sq}\\
		=&\sum\limits_{j=1}^{m}\left(u_{j}c_{j}^{m-\delta }\sum\limits_{\alpha\in C_{j}}\alpha^{r+sq}\right).
	\end{aligned}
	$$
	Note that $C_{j}=\beta_{j}\cdot U_{q+1}$ and  $c_{j}=\beta_{j}^{q+1}$, then $\alpha^{q+1}=c_{j}$ for $\alpha\in C_{j}$, furthermore,
	$$S_{r+sq}^{(H)}=\sum\limits_{j=1}^{m}\left(u_{j}c_{j}^{m-\delta }\sum\limits_{\alpha\in C_{j}}\alpha^{r+sq}\right)=\sum\limits_{j=1}^{m}\left(u_{j}c_{j}^{m-\delta +s}\sum\limits_{\alpha\in C_{j}}\alpha^{r-s}\right).$$
	Now from that $C_j$ is a coset of $U _{q+1}$, thus for any $\mu \in \mathbb{N}$, 
	$$
	\sum _{\alpha \in C_j} \alpha ^{\mu} = \begin{cases}
		0,               & \text{ if } q+1 \nmid \mu; \\
		(q+1)\beta _{j} ^{\mu}, & \text{ if } q+1 \mid \mu.
	\end{cases}
	$$
	Furthermore, for $0\leq r\neq s\leq k-1$, by $\delta\leq m$, we have $k\leq m+k-\delta$, thus
	$$|r-s|\leq k-1\leq m+k-\delta-1\leq q-2-(k-\delta)+k-\delta\leq q-2<q+1,$$ 
	and so $q+1|r-s$ if and only $r =s$. Hence, $S _{r+qs} ^{(H)}=0$ for $0\leq r\neq s\leq k-1$.
	
	Next, for  $0\leq r=s\leq k-1$, we have $S_{s(1+q)}^{(H)}=(q+1)\sum\limits_{j=1}^{m}u_{j}c_{j}^{m-\delta +s}.$ Now, we divide the following two cases.
	
	{\bf Case 1.} If $\delta=1$, then $$S_{s(1+q)}^{(H)}=(q+1)\sum\limits_{j=1}^{m}u_{j}c_{j}^{m-1+s}=(q+1)S_{s}(c_{1},\cdots,c_{m})\triangleq (q+1)S_{s}\left(0<r=s\leq k-1\right).$$
	
	And so
	\begin{equation}\label{GGH1}
		\begin{aligned}
			&\boldsymbol{G}_{\boldsymbol{v},k,t}\boldsymbol{G}_{\boldsymbol{v},k,t}^{T_{h}}\\
			=&\begin{pmatrix}
				v_{n+1}^{1+q}+S_{0}^{(q+1)}&0&\cdots&0&v_{n+1}^{1+q}\eta^{q}&0&\cdots&0\\
				0&S_{1}^{(q+1)}&\cdots&0&0&0&\cdots&0\\
				\vdots&\vdots& &\vdots&\vdots&\vdots& &\vdots\\
				0&0&\cdots&S_{t-1}^{(q+1)}&0&0&\cdots&0\\
				v_{n+1}^{1+q}\eta&0&\cdots&0&v_{n+1}^{1+q}\eta^{1+q}+S_{t}^{(q+1)}&0&\cdots&0\\
				0&0&\cdots&0&0&S_{t+1}^{(q+1)}&\cdots&0\\
				\vdots&\vdots& &\vdots&\vdots&\vdots& &\vdots\\  
				0&0&\cdots&0&0&0&\cdots&S_{k-1}^{(q+1)}\\
			\end{pmatrix},
		\end{aligned}
	\end{equation}
	where $S_{i}^{(q+1)}=(q+1)S_{i}$.
	Thus $$\rank\left(\boldsymbol{G}_{\boldsymbol{v},k,t}\boldsymbol{G}_{\boldsymbol{v},k,t}^{T_{h}}\right)=k-2+\rank\begin{pmatrix}
		v_{n+1}^{1+q}+S_{0}^{(q+1)}&v_{n+1}^{1+q}\eta^{q}\\
		v_{n+1}^{1+q}\eta&v_{n+1}^{1+q}\eta^{q}+S_{t}^{(q+1)}\\
	\end{pmatrix}=k-1\ \text{or}\ k,$$
	i.e.,
	$$
	\dim\left(\ehull\left(\mathrm{MGRS}_{k+1,k}\left(\boldsymbol{\alpha},\boldsymbol{v},\eta,t\right)\right)\right)=0\ \text{or}\ 1.
	$$
	
	{\bf Case 2} If $\delta \geq 2$, then for $S_{s(1+q)}^{(H)}$, by Lemma \ref{usapowersum}, we know that if $m-\delta +s\leq m-2$, i.e., $s\leq \delta -2$, then $$S_{s(1+q)}^{(H)}=\sum\limits_{j=1}^{m}u_{j}c_{j}^{m-\delta +s}=0$$ for $0\leq r=s\leq \delta -2$; if $m-\delta +s>m-2$, i.e., $s=\delta -2+\ell(1\leq\ell\leq k-\delta +1)$, then 
	for $\delta -2<r=s\leq k-1$, we have
	$$
	\begin{aligned}
		S_{s(1+q)}^{(H)}=&(q+1)\sum\limits_{j=1}^{m}u_{j}c_{j}^{m-\delta+s}\\
		=&(q+1)\sum\limits_{j=1}^{m}u_{j}c_{j}^{m-2+\ell}\\
		=&(q+1)S_{\ell-1}(c_{1},\cdots,c_{m})\\
		\triangleq& (q+1)S_{\ell-1}.
	\end{aligned}
	$$
	
	And so
	$$
	\begin{aligned}
		\boldsymbol{G}_{\boldsymbol{v},k,t}\boldsymbol{G}_{\boldsymbol{v},k,t}^{T_{h}}=&\begin{pmatrix}
			S_{0}^{(H)}+v_{n+1}^{1+q}&0&\cdots&0&v_{n+1}^{1+q}\eta^q&0&\cdots&0\\
			0&S_{1+q}^{(H)}&\cdots&0&0&0&\cdots&0\\
			\vdots&\vdots& &\vdots&\vdots&\vdots& &\vdots\\
			0&0&\cdots&S_{(t-1)(1+q)}^{(H)}&0&0&\cdots&0\\
			v_{n+1}^{1+q}\eta&0&\cdots&0&S_{t(1+q)}^{(H)}+v_{n+1}^{1+q}\eta^{1+q}&0&\cdots&0\\
			0&0&\cdots&0&0&S_{(t+1)(1+q)}^{(H)}&\cdots&0\\
			\vdots&\vdots& &\vdots&\vdots&\vdots& &\vdots\\ 
			0&0&\cdots&0&0&0&\cdots&S_{(k-1)(1+q)}^{(H)}
		\end{pmatrix}\\
		=&(q+1)\begin{pmatrix}
			0&0&\cdots&0&0&0&\cdots&0\\
			0&0&\cdots&0&0&0&\cdots&0\\
			\vdots&\vdots& &\vdots&\vdots&\vdots& &\vdots\\
			0&0&\cdots&0&0&0&\cdots&0\\
			0&0&\cdots&0&S_{0}&0&\cdots&0\\
			0&0&\cdots&0&0&S_{1}&\cdots&0\\
			\vdots&\vdots& &\vdots&\vdots&\vdots& &\\ 
			0&0&\cdots&0&0&0&\cdots&S_{k-\delta}
		\end{pmatrix}+\begin{pmatrix}
			v_{n+1}^{1+q}&0&\cdots&0&v_{n+1}^{1+q}\eta^q&0&\cdots&0\\
			0&0&\cdots&0&0&0&\cdots&0\\
			\vdots&\vdots& &\vdots&\vdots&\vdots& &\vdots\\
			0&0&\cdots&0&0&0&\cdots&0\\
			v_{n+1}^{1+q}\eta&0&\cdots&0&v_{n+1}^{1+q}\eta^{1+q}&0&\cdots&0\\
			0&0&\cdots&0&0&0&\cdots&0\\
			\vdots&\vdots& &\vdots&\vdots&\vdots& &\vdots\\
			0&0&\cdots&0&0&0&\cdots&0\\
			0&0&\cdots&0&0&0&\cdots&0\\
			0&0&\cdots&0&0&0&\cdots&0
		\end{pmatrix}
	\end{aligned}
	$$
	
	Next, for the case $\delta \geq 2$, according to the value of $t$, we divide it into the following two cases. 
	
	{\bf Case 2.1.} If $\delta \geq 3$ and $1\leq t\leq \delta-2$, then 
	$$ 
		\boldsymbol{G}_{\boldsymbol{v},k,t}\boldsymbol{G}_{\boldsymbol{v},k,t}^{T_{h}}=\begin{pmatrix}
			v_{n+1}^{1+q}&0&\cdots&0&v_{n+1}^{1+q}\eta^{q}&0&\cdots&0&0&\cdots&0\\
			0&0&\cdots&0&0&0&\cdots&0&0&\cdots&0\\
			\vdots&\vdots& &\vdots&\vdots&\vdots& &\vdots&\vdots& &\vdots\\
			0&0&\cdots&0&0&0&\cdots&0&0&\cdots&0\\
			v_{n+1}^{1+q}\eta&0&\cdots&0&v_{n+1}^{1+q}\eta^{1+q}&0&\cdots&0&0&\cdots&0\\
			0&0&\cdots&0&0&0&\cdots&0&0&\cdots&0\\
			\vdots&\vdots& &\vdots&\vdots&\vdots& &\vdots&\vdots& &\vdots\\
			0&0&\cdots&0&0&0&\cdots&0&0&\cdots&0\\ 
			0&0&\cdots&0&0&0&\cdots&0&S_{0}^{(q+1)}&\cdots&0\\ 
			\vdots&\vdots& &\vdots&\vdots&\vdots& &\vdots&\vdots& &\vdots\\
			0&0&\cdots&0&0&0&\cdots&0&0&\cdots&S_{k-\delta}^{(q+1)}\\
		\end{pmatrix}
$$
Through elementary row operations, we have
$$
\rank\left(\boldsymbol{G}_{\boldsymbol{v},k,t}\boldsymbol{G}_{\boldsymbol{v},k,t}^{T_{h}}\right)=\rank\begin{pmatrix}
			v_{n+1}^{1+q}&0&0&\cdots&0&0&\cdots&0\\
			0&0&0&\cdots&0&0&\cdots&0\\
			0&0&0&\cdots&0&0&\cdots&0\\ 
			\vdots&\vdots&\vdots& &\vdots&\vdots& &\vdots\\
			0&0&0&\cdots&0&0&\cdots&0\\ 
			0&0&0&\cdots&0&S_{0}^{(q+1)}&\cdots&0\\ 
			\vdots&\vdots&\vdots& &\vdots&\vdots& &\vdots\\
			0&0&0&\cdots&0&0&\cdots&S_{k-\delta}^{(q+1)}\\
		\end{pmatrix}, 
	$$
	where $S_{i}^{(q+1)}=(q+1)S_{i}$.
	And so
	$$
	\dim\left(\ehull\left(\mathrm{MGRS}_{k+1,k}\left(\boldsymbol{\alpha},\boldsymbol{v},\eta,t\right)\right)\right)=k-(k-\delta+2)=\delta-2.
	$$
	
	{\bf Case 2.2.} If $\delta \geq 2$ and  $t=\delta-2+i(1\leq i\leq k-\delta+1)$, then 
	$$
\boldsymbol{G}_{\boldsymbol{v},k,t}\boldsymbol{G}_{\boldsymbol{v},k,t}^{T_{h}}=\begin{pmatrix}
			v_{n+1}^{1+q}&0&\cdots&0&0&\cdots&0&v_{n+1}^{1+q}\eta^{q}&0&\cdots&0\\
			0&0&\cdots&0&0&\cdots&0&0&0&\cdots&0\\
			\vdots&\vdots& &\vdots&\vdots& &\vdots&\vdots&\vdots& &\vdots\\
			0&0&\cdots&0&0&\cdots&0&0&0&\cdots&0\\
			0&0&\cdots&0&S_{0}^{(q+1)}&\cdots&0&0&0&\cdots&0\\
			\vdots&\vdots& &\vdots&\vdots& &\vdots&\vdots&\vdots& &\vdots\\
			0&0&\cdots&0&0&\cdots&S_{\delta-2+i-2}^{(q+1)}&0&0&\cdots&0\\ 
			v_{n+1}^{1+q}\eta&0&\cdots&0&0&\cdots&0&v_{n+1}^{1+q}\eta^{1+q}+S_{\delta-2+i-1}^{(q+1)}&0&\cdots&0\\
			0&0&\cdots&0&0&\cdots&0&0&S_{\delta-2+i}^{(q+1)}&\cdots&0\\ 
			\vdots&\vdots& &\vdots&\vdots& &\vdots&\vdots&\vdots& &\vdots\\
			0&0&\cdots&0&0&\cdots&0&0&0&\cdots&S_{k-\delta}^{(q+1)}\\
		\end{pmatrix}$$
Through elementary row operations, we have
$$
\rank\left(\boldsymbol{G}_{\boldsymbol{v},k,t}\boldsymbol{G}_{\boldsymbol{v},k,t}^{T_{h}}\right)=\rank\begin{pmatrix}
			v_{n+1}^{1+q}&0&0&\cdots&0&0&\cdots&0\\
			0&0&0&\cdots&0&0&\cdots&0\\
			0&0&0&\cdots&0&0&\cdots&0\\ 
			\vdots&\vdots&\vdots& &\vdots&\vdots& &\vdots\\
			0&0&0&\cdots&0&0&\cdots&0\\ 
			0&0&0&\cdots&0&S_{0}^{(q+1)}&\cdots&0\\ 
			\vdots&\vdots&\vdots& &\vdots&\vdots& &\vdots\\
			0&0&0&\cdots&0&0&\cdots&S_{k-\delta}^{(q+1)}\\
		\end{pmatrix}, 
	$$
	where $S_{i}^{(q+1)}=(q+1)S_{i}$.
	And so
	$$
	\dim\left(\ehull\left(\mathrm{MGRS}_{k+1,k}\left(\boldsymbol{\alpha},\boldsymbol{v},\eta,t\right)\right)\right)=k-(k-\delta+2)=\delta-2.
	$$ 
	
	This completes the proof of Theorem \ref{MGRSHhull}.
	
	$\hfill\Box$
	
	\section{The non-equivalence between MGRS codes and Elliptic Curve codes}\label{sec7} 
	
	In this section, by the Schur product, we prove that any NMDS MGRS code is not linearly equivalent to any  linear code $\mathcal{C}(P, G, E)$ of the elliptic-curve type from an
	elliptic curve $E$, where $P=\left\{P_{1},\ldots,P_{n}\right\},$ $G = kP_{0}.$
	
	Firstly, we recall the definition of the Schur product and a key lemma.
	\begin{definition}\label{schurproduct}{\rm(\cite{ZhouHYNMDS2025}, Definitions II.1-2)}
		For $\boldsymbol{x} = (x_1, \ldots, x_n)$ and $\boldsymbol{y} = (y_1, \ldots, y_n) \in \mathbb{F}_q^n$, the Schur product between $\boldsymbol{x}$ and $\boldsymbol{y}$ is defined as
		\[
		\boldsymbol{x} \star \boldsymbol{y} := (x_1 y_1, \ldots, x_n y_n).
		\]
		The Schur product of two $q$-ary codes $\mathcal{C}_1$ and $\mathcal{C}_2$ with length $n$ is defined as
		\[
		\mathcal{C}_1 \star \mathcal{C}_2 = \langle \boldsymbol{c}_1 \star \boldsymbol{c}_2 \mid \boldsymbol{c}_1 \in \mathcal{C}_1, \boldsymbol{c}_2 \in \mathcal{C}_2 \rangle.
		\]
		Especially, for any code $\mathcal{C}$, $\mathcal{C}^2\triangleq\mathcal{C} \star \mathcal{C}$ is called the Schur square of $\mathcal{C}$.	
	\end{definition}
	\begin{lemma}\label{noneqAG}{\rm(\cite{ChennonRSAG}, Theorem 41)}
		The dimension of the Schur square of a one-point AG code $\mathcal{C}(P, G, E)$ from an elliptic
		curve $E$ is $2k$ when $6\le 2k\le n.$ 
	\end{lemma}
	
	Next, basing on Definition \ref{schurproduct} and Lemma \ref{noneqAG}, we can obtain the following
	\begin{theorem}\label{MGRSAG}
		When $6\le 2k\le n,$ any MGRS code is not linearly equivalent to
		any elliptic curve code $\mathcal{C}(P, G, E)$.
	\end{theorem}
	{\bf Proof}. By Definition \ref{MGRSdefinition} and the definition of the monomially equivalent, we only focus on $\mathrm{MGRS}_{n,k}\left(\boldsymbol{\alpha},\boldsymbol{1},\eta,t\right)$ which is generated by the following
	vectors
	$$\boldsymbol{\alpha}_{0}=\left(1,\ldots,1,1\right),\boldsymbol{\alpha}_{i}=\left(\alpha_{1}^{i},\ldots,\alpha_{n-1}^{i},0\right),\boldsymbol{\alpha}_{t}=\left(\alpha_{1}^{t},\ldots,\alpha_{n-1}^{t},\eta\right)(i\in\left\{1,\ldots,t-1,t+2,\ldots,k-1\right\}).$$
	Thus, by Definition \ref{schurproduct},  $\mathrm{MGRS}_{n,k}^{2}\left(\boldsymbol{\alpha},\boldsymbol{1},\eta,t\right)$ is generated by the following
	vectors
	$$\boldsymbol{\alpha}_{0}\star\boldsymbol{\alpha}_{0}=\left(1,\ldots,1,1\right),$$
	$$\boldsymbol{\alpha}_{0}\star\boldsymbol{\alpha}_{i}=\left(\alpha_{1}^{i},\ldots,\alpha_{n-1}^{i},0\right)(i\in\left\{1,\ldots,t-1,t+2,\ldots,k-1\right\}),$$
	$$\boldsymbol{\alpha}_{0}\star\boldsymbol{\alpha}_{t}=\left(\alpha_{1}^{t},\ldots,\alpha_{n-1}^{t},\eta\right),$$
	$$\boldsymbol{\alpha}_{i}\star\boldsymbol{\alpha}_{j}=\left(\alpha_{1}^{i+j},\ldots,\alpha_{n-1}^{i+j},0\right)(i,j\in\left\{1,\ldots,t-1,t+2,\ldots,k-1\right\}),$$
	$$\boldsymbol{\alpha}_{i}\star\boldsymbol{\alpha}_{t}=\left(\alpha_{1}^{i+t},\ldots,\alpha_{n-1}^{i+t},0\right)(i\in\left\{1,\ldots,t-1,t+2,\ldots,k-1\right\}),$$
	$$\boldsymbol{\alpha}_{t}\star\boldsymbol{\alpha}_{t}=\left(\alpha_{1}^{2t},\ldots,\alpha_{n-1}^{2t},\eta^2\right),$$
	i.e.,
	$$\mathrm{MGRS}_{n,k}^{2}\left(\boldsymbol{\alpha},\boldsymbol{1},\eta,t\right)=\left\langle\boldsymbol{\beta}_{1},\ldots,\boldsymbol{\beta}_{2k-2},\boldsymbol{\beta}_{2k-1},\boldsymbol{\beta}_{2k},\boldsymbol{\beta}_{2k+1}\right\rangle,$$
	where
	$$\boldsymbol{\beta}_{i}=\begin{cases}
		\left(\alpha_{1}^{i},\ldots,\alpha_{n-1}^{i},0\right),&\ \text{if}\ 1\leq i\leq 2k-2;\\
		\left(1,\ldots,1,1\right),&\ \text{if}\ i=2k-1;\\
		\left(\alpha_{1}^{t},\ldots,\alpha_{n-1}^{t},\eta\right),&\ \text{if}\ i=2k;\\
		\left(\alpha_{1}^{2t},\ldots,\alpha_{n-1}^{2t},\eta^2\right),&\ \text{if}\ i=2k+1,\\
	\end{cases}$$
	namely, 
	$$\mathrm{MGRS}_{n,k}^{2}\left(\boldsymbol{\alpha},\boldsymbol{1},\eta,t\right)=U+W,$$
	where
	$$U=\left\langle\boldsymbol{\beta}_{1},\ldots,\boldsymbol{\beta}_{2k-2}\right\rangle, W=\left\langle\boldsymbol{\beta}_{2k-1},\boldsymbol{\beta}_{2k},\boldsymbol{\beta}_{2k+1}\right\rangle.$$
	
	By $2k\leq n$ and $\alpha_{i}\neq\alpha_{j}(\i\neq j)$, we have $\dim\left(U\right)=2k-2$ and $\dim\left(W\right)=3$. Thus
	$$\begin{aligned}
		\dim\left(\mathrm{MGRS}_{n,k}^{2}\left(\boldsymbol{\alpha},\boldsymbol{1},\eta,t\right)\right)=&\dim\left(U+W\right)\\
		=&\dim\left(U\right)+\dim\left(W\right)-\dim\left(U\cap W\right)\\
		=&2k+1-\dim\left(U\cap W\right).
	\end{aligned}$$
	
	Next, we prove $\dim\left(U\cap W\right)=0,$ i.e., $U\cap W=\left\{\boldsymbol{0}\right\}.$
	
	For any $\boldsymbol{x}=\left(x_{1},\ldots,x_{n-1},x_{n}\right)\in U\cap W$. By $\boldsymbol{x}\in W$, we know that there exist scalars $a_1, a_2, a_3\in\mathbb{F}_{q}$ such that
	$$\boldsymbol{x}=a_{1}\boldsymbol{\beta}_{2k-1}+a_{2}\boldsymbol{\beta}_{2k}+a_{3}\boldsymbol{\beta}_{2k+1}.$$
	Furthermore, by $\boldsymbol{x}\in U$, we know that there exist $b_1, \ldots, b_{2k-2}\in\mathbb{F}_{q}$  such that
	$$\boldsymbol{x}=b_{1}\boldsymbol{\beta}_{1}+\ldots+b_{2k-2}\boldsymbol{\beta}_{2k-2}.$$
	And so, 
	$$a_{1}\boldsymbol{\beta}_{2k-1}+a_{2}\boldsymbol{\beta}_{2k}+a_{3}\boldsymbol{\beta}_{2k+1}=b_{1}\boldsymbol{\beta}_{1}+\ldots+b_{2k-2}\boldsymbol{\beta}_{2k-2},$$
	i.e.,
	$$
	\begin{aligned}
		&a_{1}\left(1,\ldots,1,1\right)+a_{2}\left(\alpha_{1}^{t},\ldots,\alpha_{n-1}^{t},\eta\right)+a_{3}\left(\alpha_{1}^{2t},\ldots,\alpha_{n-1}^{2t},\eta^{2}\right)\\
		=&b_{1}\left(\alpha_{1}^{1},\ldots,\alpha_{n-1}^{1},0\right)+\ldots+b_{2k-2}\left(\alpha_{1}^{2k-2},\ldots,\alpha_{n-1}^{2k-2},0\right).
	\end{aligned}
	$$
	By comparing the first $n-1$ components, we have the homogeneous linear system
	\begin{equation}\label{py}
		\boldsymbol{P}\boldsymbol{y}^{T}
		=\boldsymbol{0},
	\end{equation}
	where $\boldsymbol{y}=\left(	b_{2k-2},\ldots,b_{2t+1},b_{2t}-a_3,b_{2t-1},\ldots,b_{t+1},b_t-a_2,b_{t-1},\ldots,b_1,-a_1\right)$ and 
	$$\boldsymbol{P}=
	\begin{pmatrix}
		\alpha_1^{2k-2} & \dots & \alpha_1^{2t+1} & \alpha_1^{2t} & \alpha_1^{2t-1} & \dots & \alpha_1^{t+1} & \alpha_1^t & \alpha_1^{t-1} & \dots & \alpha_1^1 & 1\\
		\alpha_2^{2k-2} & \dots & \alpha_2^{2t+1} & \alpha_2^{2t} & \alpha_2^{2t-1} & \dots & \alpha_2^{t+1} & \alpha_2^t & \alpha_2^{t-1} & \dots & \alpha_2^1 & 1\\
		\vdots & \ddots & \vdots & \vdots & \vdots & \ddots & \vdots & \vdots & \vdots & \ddots & \vdots & \vdots\\
		\alpha_{n-1}^{2k-2} & \dots & \alpha_{n-1}^{2t+1} & \alpha_{n-1}^{2t} & \alpha_{n-1}^{2t-1} & \dots & \alpha_{n-1}^{t+1} & \alpha_{n-1}^t & \alpha_{n-1}^{t-1} & \dots & \alpha_{n-1}^1 & 1
	\end{pmatrix}.$$
	It is easy to know that the dimension of the solution space for this homogeneous linear system \eqref{py} is equal to $(2k-2)-(2k-2)=0$, i.e., $\boldsymbol{y}=\boldsymbol{0}$,  it means that $U\cap W=\left\{\boldsymbol{0}\right\},$ thus 
	$$\dim\left(\mathrm{MGRS}_{n,k}^{2}\left(\boldsymbol{\alpha},\boldsymbol{1},\eta,t\right)\right)=2k+1\neq 2k,$$  
	And so, by Lemma \ref{noneqAG}, this completes the proof of Theorem \ref{MGRSAG}.
	
	$\hfill\Box$  
	\section{Some examples}\label{sec8}
	
	In this section, we give some examples for Theorems  \ref{NMDSMGRSK-1}-\ref{NMDSMGRS1}, Theorems   \ref{MDSNMDSEMGRS2}-\ref{NMDSEMGRS1}, Theorem \ref{EHullMGRS1} and Theorem \ref{MGRSHhull}, where Example \ref{MGRSMDSNMDS} is for Theorems \ref{NMDSMGRSK-1}-\ref{NMDSMGRS1}, Example \ref{EMGRSMDSNMDS} is for Theorem \ref{NMDSEMGRS1}, 
	Example \ref{MDSNMDSEMGRS1example} is for Theorem \ref{MGRSk-1weight}, Example \ref{MDSNMDSEMGRS2example} is for Theorem \ref{EMGRSk-1weight}, Example \ref{EHullMGRS1example} is for Theorem \ref{EHullMGRS1}, Example \ref{MGRSHhullexample} is for Theorem \ref{MGRSHhull}, respectively.
	
	\subsection{Two examples for Theorems \ref{NMDSMGRSK-1}-\ref{NMDSMGRS1} and Theorem \ref{NMDSEMGRS1}}
	\begin{example}\label{MGRSMDSNMDS}
		Let $\left(q,k,n\right)=(27,5,8), \mathbb{F}_{27}^{*}=\langle\omega\rangle$ and
		$$H=\left\{\alpha_{1},\ldots,\alpha_{n-1}\right\}=\left\{\omega,\omega^2,\omega^3,\omega^4,\omega^5,\omega^6,\omega^7\right\}\subseteq\mathbb{F}_{27}^{*}.$$
		By directly calculating, we have
		$$\widetilde{H}_{1}=\left\{(-1)^{k}\prod\limits_{\alpha_{i}\in \mathcal{H}}\alpha_{i}:  \mathcal{H}\subseteq H, \#\mathcal{H}=k-1\right\}=\left\{1, \omega, \omega^2, \omega^3, \omega^4, \omega^5,\omega^6, \omega^7,  \omega^8, \omega^9,\omega^{23}, \omega^{24}, \omega^{25}\right\}$$
		and
		$$
		\widetilde{H}_{2}=\left\{\sum\limits_{\alpha_{i}\in \mathcal{H}}\alpha_{i}^{-1}:  \mathcal{H}\subseteq H, \#\mathcal{H}=k-1\right\}=\left\{\omega^{i}|0\leq i\leq 25\right\}\cup\left\{0\right\}.$$
		It is easy to know that $\widetilde{H}_{1}\neq \mathbb{F}_{27}^{*}$ and $\widetilde{H}_{2}=\mathbb{F}_{27}.$ Then, we have the following three statements. 
		
		$(1)$ By taking $\eta=\omega^{9}\in\widetilde{H}_{1}$, the corresponding MGRS code has the generator matrix
		$$\begin{pmatrix}
			1&1&1&1&1&1&1&1\\
			\omega&\omega^{2}&\omega^{3}&\omega^{4}&\omega^{5}&\omega^{6}&\omega^{7}&0\\
			\omega^{2}&\left(\omega^{2}\right)^{2}&\left(\omega^{3}\right)^{2}&\left(\omega^{4}\right)^{2}&\left(\omega^{5}\right)^{2}&\left(\omega^{6}\right)^{2}&\left(\omega^{7}\right)^{2}&0\\ 
			\omega^{3}&\left(\omega^{2}\right)^{3}&\left(\omega^{3}\right)^{3}&\left(\omega^{4}\right)^{3}&\left(\omega^{5}\right)^{3}&\left(\omega^{6}\right)^{3}&\left(\omega^{7}\right)^{3}&0\\
			\omega^{4}&\left(\omega^{2}\right)^{4}&\left(\omega^{3}\right)^{4}&\left(\omega^{4}\right)^{4}&\left(\omega^{5}\right)^{4}&\left(\omega^{6}\right)^{4}&\left(\omega^{7}\right)^{4}&\omega^{9}
		\end{pmatrix}.$$
		Based on the Magma program, the MGRS code is an NMDS code with the parameters $[8,5,3]_{3^3}$, which is consistent with Theorem \ref{NMDSMGRSK-1}.
		
		$(2)$ By taking $\eta=\omega^{10}\in \mathbb{F}_{27}^{*}\backslash\widetilde{H}_{1}$, the corresponding MGRS code has the generator matrix
		$$\begin{pmatrix}
			1&1&1&1&1&1&1&1\\
			\omega&\omega^{2}&\omega^{3}&\omega^{4}&\omega^{5}&\omega^{6}&\omega^{7}&0\\
			\omega^{2}&\left(\omega^{2}\right)^{2}&\left(\omega^{3}\right)^{2}&\left(\omega^{4}\right)^{2}&\left(\omega^{5}\right)^{2}&\left(\omega^{6}\right)^{2}&\left(\omega^{7}\right)^{2}&0\\ 
			\omega^{3}&\left(\omega^{2}\right)^{3}&\left(\omega^{3}\right)^{3}&\left(\omega^{4}\right)^{3}&\left(\omega^{5}\right)^{3}&\left(\omega^{6}\right)^{3}&\left(\omega^{7}\right)^{3}&0\\
			\omega^{4}&\left(\omega^{2}\right)^{4}&\left(\omega^{3}\right)^{4}&\left(\omega^{4}\right)^{4}&\left(\omega^{5}\right)^{4}&\left(\omega^{6}\right)^{4}&\left(\omega^{7}\right)^{4}&\omega^{10}
		\end{pmatrix}.$$
		Based on the Magma program, the MGRS code is an MDS code with the parameters $[8,5,4]_{3^3}$, which is consistent with Theorem \ref{NMDSMGRSK-1}.
		
		$(3)$ By taking $\eta^{-1}=\omega^{4}\in\widetilde{H}_{2}$, i.e., $\eta=\omega^{22}$, the corresponding MGRS code has the generator matrix
		$$\begin{pmatrix}
			1&1&1&1&1&1&1&1\\
			\omega&\omega^{2}&\omega^{3}&\omega^{4}&\omega^{5}&\omega^{6}&\omega^{7}&\omega^{22}\\
			\omega^{2}&\left(\omega^{2}\right)^{2}&\left(\omega^{3}\right)^{2}&\left(\omega^{4}\right)^{2}&\left(\omega^{5}\right)^{2}&\left(\omega^{6}\right)^{2}&\left(\omega^{7}\right)^{2}&0\\ 
			\omega^{3}&\left(\omega^{2}\right)^{3}&\left(\omega^{3}\right)^{3}&\left(\omega^{4}\right)^{3}&\left(\omega^{5}\right)^{3}&\left(\omega^{6}\right)^{3}&\left(\omega^{7}\right)^{3}&0\\
			\omega^{4}&\left(\omega^{2}\right)^{4}&\left(\omega^{3}\right)^{4}&\left(\omega^{4}\right)^{4}&\left(\omega^{5}\right)^{4}&\left(\omega^{6}\right)^{4}&\left(\omega^{7}\right)^{4}&0
		\end{pmatrix}.$$
		Based on the Magma program, the MGRS code is an NMDS code with the parameters $[8,5,3]_{3^3}$, which is consistent with Theorem \ref{NMDSMGRS1}. 
	\end{example}
	
	\begin{example}\label{EMGRSMDSNMDS}
		Let $\left(q,k,n\right)=(27,5,8), \mathbb{F}_{27}^{*}=\langle\omega\rangle$ and
		$$H=\left\{\alpha_{1},\ldots,\alpha_{n-1}\right\}=\left\{\omega,\omega^2,\omega^3,\omega^4,\omega^5,\omega^7,\omega^8\right\}\subseteq\mathbb{F}_{7}^{*}.$$
		By directly calculating, we have 
		$$	
		\begin{aligned}
			\widetilde{H}_{2} &= \left\{\sum_{\alpha_{i} \in \mathcal{H}} \alpha_{i}^{-1} : \mathcal{H} \subseteq H,\ |\mathcal{H}|=k-2\right\} \\
			&= \left\{1, \omega, \omega^2, \omega^3, \omega^4, \omega^5, \omega^6, \omega^8, \omega^9, \omega^{10}, \omega^{11}, \omega^{12}, 2, \omega^{15}, \omega^{16}, \omega^{17}, \omega^{19}, \omega^{20}, \omega^{21}, \omega^{23}, \omega^{24}, \omega^{25}, 0 \right\}, \\
		\end{aligned}
		$$
		and
		$$	
		\begin{aligned}
			\widetilde{H}_{3} &= \left\{\sum_{\alpha_{i} \in \mathcal{H}} \alpha_{i}^{-1} : \mathcal{H} \subseteq H,\ |\mathcal{H}|=k-1\right\} \\
			&= \left\{1, \omega, \omega^2, \omega^3, \omega^5, \omega^6, \omega^8, \omega^9, \omega^{10}, \omega^{11}, \omega^{12}, 2, \omega^{15}, \omega^{17}, \omega^{18}, \omega^{19}, \omega^{20}, \omega^{21}, \omega^{22}, \omega^{23}, \omega^{24}, \omega^{25}, 0\right\}.
		\end{aligned}
		$$
		It is easy to know that $\widetilde{H}_{4}=\widetilde{H}_{2}\cup\widetilde{H}_{3}\neq \mathbb{F}_{27}^{*}$ and $\mathbb{F}_{27}\backslash\widetilde{H}_{4}=\left\{\omega^{7},\omega^{14}\right\}.$ Then, we have the following two statements. 
		
		$(1)$ By taking $\eta^{-1}=\omega^{9}\in\widetilde{H}_{4}$, i.e., $\eta=\omega^{17}$, the corresponding MGRS code has the generator matrix
		$$\begin{pmatrix}
			1&1&1&1&1&1&1&1&0\\
			\omega&\omega^{2}&\omega^{3}&\omega^{4}&\omega^{5}&\omega^{6}&\omega^{7}&\omega^{17}&0\\
			\omega^{2}&\left(\omega^{2}\right)^{2}&\left(\omega^{3}\right)^{2}&\left(\omega^{4}\right)^{2}&\left(\omega^{5}\right)^{2}&\left(\omega^{6}\right)^{2}&\left(\omega^{7}\right)^{2}&0&0\\ 
			\omega^{3}&\left(\omega^{2}\right)^{3}&\left(\omega^{3}\right)^{3}&\left(\omega^{4}\right)^{3}&\left(\omega^{5}\right)^{3}&\left(\omega^{6}\right)^{3}&\left(\omega^{7}\right)^{3}&0&0\\
			\omega^{4}&\left(\omega^{2}\right)^{4}&\left(\omega^{3}\right)^{4}&\left(\omega^{4}\right)^{4}&\left(\omega^{5}\right)^{4}&\left(\omega^{6}\right)^{4}&\left(\omega^{7}\right)^{4}&0&1
		\end{pmatrix}.$$
		Based on the Magma program, the MGRS code is an NMDS code with the parameters $[9,5,4]_{3^3}$, which is consistent with Theorem \ref{NMDSEMGRS1}. 
		
		$(2)$ By taking $\eta^{-1}=\omega^{14}\in\mathbb{F}_{27}\backslash\widetilde{H}_{4}$, i.e., $\eta=\omega^{12}$, the corresponding MGRS code has the generator matrix
		$$\begin{pmatrix}
			1&1&1&1&1&1&1&1&0\\
			\omega&\omega^{2}&\omega^{3}&\omega^{4}&\omega^{5}&\omega^{6}&\omega^{7}&\omega^{12}&0\\
			\omega^{2}&\left(\omega^{2}\right)^{2}&\left(\omega^{3}\right)^{2}&\left(\omega^{4}\right)^{2}&\left(\omega^{5}\right)^{2}&\left(\omega^{6}\right)^{2}&\left(\omega^{7}\right)^{2}&0&0\\ 
			\omega^{3}&\left(\omega^{2}\right)^{3}&\left(\omega^{3}\right)^{3}&\left(\omega^{4}\right)^{3}&\left(\omega^{5}\right)^{3}&\left(\omega^{6}\right)^{3}&\left(\omega^{7}\right)^{3}&0&0\\
			\omega^{4}&\left(\omega^{2}\right)^{4}&\left(\omega^{3}\right)^{4}&\left(\omega^{4}\right)^{4}&\left(\omega^{5}\right)^{4}&\left(\omega^{6}\right)^{4}&\left(\omega^{7}\right)^{4}&0&1
		\end{pmatrix}.$$
		Based on the Magma program, the MGRS code is an MDS code with the parameters $[9,5,5]_{3^3}$, which is consistent with Theorem \ref{NMDSEMGRS1}. 
	\end{example}

	\subsection{Two examples for Theorem \ref{MGRSk-1weight} and Theorem \ref{EMGRSk-1weight}}
	\begin{example}\label{MDSNMDSEMGRS1example}
		Let $\left(q,k,n\right)=(7,5,7), \mathbb{F}_{7}^{*}=\langle\omega\rangle, \eta=\omega^{2}$ and $\alpha_{i}=\omega^{i}(i=1,2,\ldots,6)$, then 
		$$\begin{pmatrix}
			1&1&1&1&1&1&1\\
			\omega&\omega^{2}&\omega^{3}&\omega^{4}&\omega^{5}&\omega^{6}&0\\
			\omega^{2}&\left(\omega^{2}\right)^{2}&\left(\omega^{3}\right)^{2}&\left(\omega^{4}\right)^{2}&\left(\omega^{5}\right)^{2}&\left(\omega^{6}\right)^{2}&0\\ 
			\omega^{3}&\left(\omega^{2}\right)^{3}&\left(\omega^{3}\right)^{3}&\left(\omega^{4}\right)^{3}&\left(\omega^{5}\right)^{3}&\left(\omega^{6}\right)^{3}&0\\
			\omega^{4}&\left(\omega^{2}\right)^{4}&\left(\omega^{3}\right)^{4}&\left(\omega^{4}\right)^{4}&\left(\omega^{5}\right)^{4}&\left(\omega^{6}\right)^{4}&\omega^{2}
		\end{pmatrix}$$ 
		is the generator matrix of $\mathrm{MGRS}_{7,5}\left(\boldsymbol{\alpha},\boldsymbol{1},\omega^{2},4\right)$. Based on the Magma program, the  $\mathrm{MGRS}_{7,5}\left(\boldsymbol{\alpha},\boldsymbol{1},\omega^{2},4\right)$ is NMDS with the parameters $[7,5,2]$ and the weight enumerator 
		\begin{equation}\label{NMDSMGRSweight}
			A(x)=1+12x^2+150x^3+960x^4+3282x^{5}+6696x^{6}+5706x^{7}.
		\end{equation}
		By taking $\left(q,k,t\right)=(7,5,2)$ in Theorem \ref{MGRSk-1weight} and directly calculating, we have 
		$$
		\begin{aligned}
			A_{2}=&\sum_{r\mid \gcd(6,4)}(-1)^{4+\frac{4}{r}}\binom{\frac{6}{r}}{\frac{4}{r}}\sum_{d\mid (r,-13)} \mu\left(\frac{r}{d}\right)d \\
			=&(-1)^{4+\frac{4}{1}}\binom{\frac{6}{1}}{\frac{4}{1}}\sum_{d\mid (1,-13)} \mu\left(\frac{1}{d}\right)d +(-1)^{4+\frac{4}{2}}\binom{\frac{6}{2}}{\frac{4}{2}}\sum_{d\mid (2,-13)} \mu\left(\frac{2}{d}\right)d\\
			=&\binom{6}{4}\mu\left(1\right) +\binom{3}{2}\mu\left(2\right)\\
			=&15-3=12,
		\end{aligned}
		$$
		which is consistent with \eqref{NMDSMGRSweight} output by the Magma program.
	\end{example} 
	
	\begin{example}\label{MDSNMDSEMGRS2example}
		Let $\left(q,k,n\right)=(7,5,7), \mathbb{F}_{7}^{*}=\langle\omega\rangle, \eta=\omega^{2}$, and $\alpha_{i}=\omega^{i}(i=1,2,\ldots,6)$, then 
		$$\begin{pmatrix}
			1&1&1&1&1&1&1&0\\
			\omega&\omega^{2}&\omega^{3}&\omega^{4}&\omega^{5}&\omega^{6}&0&0\\
			\omega^{2}&\left(\omega^{2}\right)^{2}&\left(\omega^{3}\right)^{2}&\left(\omega^{4}\right)^{2}&\left(\omega^{5}\right)^{2}&\left(\omega^{6}\right)^{2}&0&0\\ 
			\omega^{3}&\left(\omega^{2}\right)^{3}&\left(\omega^{3}\right)^{3}&\left(\omega^{4}\right)^{3}&\left(\omega^{5}\right)^{3}&\left(\omega^{6}\right)^{3}&0&0\\
			\omega^{4}&\left(\omega^{2}\right)^{4}&\left(\omega^{3}\right)^{4}&\left(\omega^{4}\right)^{4}&\left(\omega^{5}\right)^{4}&\left(\omega^{6}\right)^{4}&\omega^{2}&1
		\end{pmatrix}$$ 
		is the generator matrix of the code $\mathrm{EMGRS}_{8,5}\left(\boldsymbol{\alpha},\boldsymbol{1},\omega^{2},4\right)$. Based on the Magma program, $\mathrm{EMGRS}_{8,5}\left(\boldsymbol{\alpha},\boldsymbol{1},\omega^{2},4\right)$ is NMDS with the parameters $[8,5,3]$ and the weight enumerator 
		\begin{equation}\label{NMDSEMGRSweight}
			A(x)=1+12x^3+360x^4+1128x^5+3912x^{6}+6492x^{7}+4902x^{8}.
		\end{equation}
		By taking $\left(q,k,t\right)=(7,5,2)$ in Theorem \ref{EMGRSk-1weight} and directly calculating, it's easy to know that the corresponding result is consistent with \eqref{NMDSEMGRSweight} output by the Magma program.
	\end{example}
	\subsection{An example for Section Theorem \ref{EHullMGRS1}}
	Owing to the limited space, we only present illustrative the example for Theorem \ref{EHullMGRS1} in this subsection. The examples for Theorems \ref{EHullMGRS2}–\ref{EHullMGRS4} can be similarly constructed and verified via the Magma program.
	\begin{example}\label{EHullMGRS1example}
		Let $p=3$, $q=3^4,\delta=2,\mathbb{F}_{q}^{*}=\langle\gamma\rangle.$
		
		$(1)$ By taking $k=4$, $t=2$ and $\eta=\gamma^{4}$. It's easy to get  $k=2t$, $p\nmid k+1$,
		$$\gamma^{\delta k}k+\eta^2=4\gamma^{8}+\gamma^8=5\gamma^8\in\mathbb{F}_{q}^{*},$$
		and
		$$\gamma^{\delta k}k+\frac{k\eta^2}{k+1}=4\gamma^{8}+\frac{4\gamma^8}{5}=4\gamma^{8}+\frac{4\gamma^8}{2}=6\gamma^{8}=0.$$ And then $\alpha_{i}=\gamma^{\frac{q-1}{k}i}=\gamma^{20i}$ and $$\boldsymbol{\alpha}=\left(\gamma^{20+2},\gamma^{40+2},\gamma^{60+2},\gamma^{80+2}\right).$$ The corresponding GRL code $\mathcal{C}_{1}$ has the following generator matrix
		$$\begin{pmatrix}
			1&1&1&1&1\\
			\gamma^{22}&\gamma^{42}&\gamma^{62}&\gamma^{82}&0\\ 
			\left(\gamma^{22}\right)^2&\left(\gamma^{42}\right)^2&\left(\gamma^{62}\right)^2&\left(\gamma^{82}\right)^2&\gamma^{4}\\
			\left(\gamma^{22}\right)^3&\left(\gamma^{42}\right)^3&\left(\gamma^{62}\right)^3&\left(\gamma^{82}\right)^3&0
		\end{pmatrix}.$$
		Based on the Magma program, $\mathcal{C}_{1}$ is a one-dimensional hull code with the parameters $[5,4,1]_{3^4}$, which is consistent with Theorem \ref{EHullMGRS1}.
		
		$(2)$ By taking $k=5$, $t=2$ and $\eta=\gamma^{2}$. It's easy to get $k\neq 2t$ and $p\mid k+1$. 
		And then $\alpha_{i}=\gamma^{\frac{q-1}{k}i}=\gamma^{16i}$ and $$\boldsymbol{\alpha}=\left(\gamma^{16+2},\gamma^{32+2},\gamma^{48+2},\gamma^{64+2},\gamma^{80+2}\right).$$ The corresponding GRL code $\mathcal{C}_{2}$ has the following generator matrix
		$$\begin{pmatrix}
			1&1&1&1&1&1\\
			\gamma^{18}&\gamma^{34}&\gamma^{50}&\gamma^{66}&\gamma^{82}&0\\ 
			\left(\gamma^{18}\right)^2&\left(\gamma^{34}\right)^2&\left(\gamma^{50}\right)^2&\left(\gamma^{66}\right)^2&\left(\gamma^{82}\right)^2&0\\
			\left(\gamma^{18}\right)^3&\left(\gamma^{34}\right)^3&\left(\gamma^{50}\right)^3&\left(\gamma^{66}\right)^3&\left(\gamma^{82}\right)^3&\gamma^{2}\\
			\left(\gamma^{18}\right)^4&\left(\gamma^{34}\right)^4&\left(\gamma^{50}\right)^4&\left(\gamma^{66}\right)^4&\left(\gamma^{82}\right)^4&0
		\end{pmatrix}.$$
		Based on the Magma program, $\mathcal{C}_{2}$ is a one-dimensional hull MDS code with the parameters $[6,5,2]_{3^4}$, which is consistent with Theorem \ref{EHullMGRS1}.
		
		$(3)$ By taking $k=8$, $t=4$ and $\eta=\gamma$. It's easy to get $k=2t$ and $p\mid k+1$. And then $\alpha_{i}=\gamma^{\frac{q-1}{k}i}=\gamma^{10i}$ and $$\boldsymbol{\alpha}=\left(\gamma^{10+2},\gamma^{20+2},\gamma^{30+2},\gamma^{40+2},\gamma^{50+2},\gamma^{60+2},\gamma^{70+2},\gamma^{80+2}\right).$$ The corresponding GRL code $\mathcal{C}_{3}$ has the following generator matrix
		$$\begin{pmatrix}
			1&1&1&1&1&1&1&1&1\\
			\gamma^{12}&\gamma^{22}&\gamma^{32}&\gamma^{42}&\gamma^{52}&\gamma^{62}&\gamma^{72}&\gamma^{82}&0\\ 
			\left(\gamma^{12}\right)^2&\left(\gamma^{22}\right)^2&\left(\gamma^{32}\right)^2&\left(\gamma^{42}\right)^2&\left(\gamma^{52}\right)^2&\left(\gamma^{62}\right)^2&\left(\gamma^{72}\right)^2&\left(\gamma^{82}\right)^2&0\\
			\left(\gamma^{12}\right)^3&\left(\gamma^{22}\right)^3&\left(\gamma^{32}\right)^3&\left(\gamma^{42}\right)^3&\left(\gamma^{52}\right)^3&\left(\gamma^{62}\right)^3&\left(\gamma^{72}\right)^3&\left(\gamma^{82}\right)^3&0\\
			\left(\gamma^{12}\right)^4&\left(\gamma^{22}\right)^4&\left(\gamma^{32}\right)^4&\left(\gamma^{42}\right)^4&\left(\gamma^{52}\right)^4&\left(\gamma^{62}\right)^4&\left(\gamma^{72}\right)^4&\left(\gamma^{82}\right)^4&\gamma\\
			\left(\gamma^{12}\right)^5&\left(\gamma^{22}\right)^5&\left(\gamma^{32}\right)^5&\left(\gamma^{42}\right)^5&\left(\gamma^{52}\right)^5&\left(\gamma^{62}\right)^5&\left(\gamma^{72}\right)^5&\left(\gamma^{82}\right)^5&0\\
			\left(\gamma^{12}\right)^6&\left(\gamma^{22}\right)^6&\left(\gamma^{32}\right)^6&\left(\gamma^{42}\right)^6&\left(\gamma^{52}\right)^6&\left(\gamma^{62}\right)^6&\left(\gamma^{72}\right)^6&\left(\gamma^{82}\right)^6&0\\
			\left(\gamma^{12}\right)^7&\left(\gamma^{22}\right)^7&\left(\gamma^{32}\right)^7&\left(\gamma^{42}\right)^7&\left(\gamma^{52}\right)^7&\left(\gamma^{62}\right)^7&\left(\gamma^{72}\right)^7&\left(\gamma^{82}\right)^7&0
		\end{pmatrix}.$$
		Based on the Magma program, $\mathcal{C}_{3}$ is a LCD MDS code with the parameters $[9,8,2]_{3^4}$, which is consistent with Theorem \ref{EHullMGRS1}.
		
		$(4)$ By taking $k=4$, $t=2$ and $\eta=\gamma^{24}$. It's easy to get $k=2t$ and $$\gamma^{\delta k}k+\eta^2=4\gamma^{8}+\gamma^{48}=4\gamma^{8}+\gamma^{40}\cdot\gamma^{8}=3\gamma^{8}=0.$$ 
		Thus $\alpha_{i}=\gamma^{\frac{q-1}{k}i}=\gamma^{20i}$ and $$\boldsymbol{\alpha}=\left(\gamma^{20+2},\gamma^{40+2},\gamma^{60+2},\gamma^{80+2}\right).$$ The corresponding GRL code $\mathcal{C}_{4}$ has the following generator matrix
		$$\begin{pmatrix}
			1&1&1&1&1\\
			\gamma^{22}&\gamma^{42}&\gamma^{62}&\gamma^{82}&0\\ 
			\left(\gamma^{22}\right)^2&\left(\gamma^{42}\right)^2&\left(\gamma^{62}\right)^2&\left(\gamma^{82}\right)^2&\gamma^{24}\\
			\left(\gamma^{22}\right)^3&\left(\gamma^{42}\right)^3&\left(\gamma^{62}\right)^3&\left(\gamma^{82}\right)^3&0
		\end{pmatrix}.$$
		Based on the Magma program, $\mathcal{C}_{4}$ is a LCD MDS code with the parameters $[5,4,2]_{3^4}$, which is consistent with Theorem \ref{EHullMGRS1}.

		$(5)$ By taking $k=4$, $t=2$ and $\eta=\gamma$. It's easy to get $k=2t$, $p\nmid k+1$,
		$$\gamma^{\delta k}k+\eta^2=4\gamma^{8}+\gamma^2=\gamma^{16}\in\mathbb{F}_{q}^{*},$$
		and
		$$\gamma^{\delta k}k+\frac{k\eta^2}{k+1}=4\gamma^{8}+\frac{4\gamma^2}{5}=\gamma^{77}\in\mathbb{F}_{q}^{*}.$$ Thus $\alpha_{i}=\gamma^{\frac{q-1}{k}i}=\gamma^{20i}$ and $$\boldsymbol{\alpha}=\left(\gamma^{20+2},\gamma^{40+2},\gamma^{60+2},\gamma^{80+2}\right).$$ The corresponding GRL code $\mathcal{C}_{5}$ has the following generator matrix
		$$\begin{pmatrix}
			1&1&1&1&1\\
			\gamma^{22}&\gamma^{42}&\gamma^{62}&\gamma^{82}&0\\ 
			\left(\gamma^{22}\right)^2&\left(\gamma^{42}\right)^2&\left(\gamma^{62}\right)^2&\left(\gamma^{82}\right)^2&\gamma\\
			\left(\gamma^{22}\right)^3&\left(\gamma^{42}\right)^3&\left(\gamma^{62}\right)^3&\left(\gamma^{82}\right)^3&0
		\end{pmatrix}.$$
		Based on the Magma program, $\mathcal{C}_{5}$ is a LCD NMDS code with the parameters $[5,4,2]_{3^4}$, which is consistent with Theorem \ref{EHullMGRS1}.
		
		$(6)$ By taking $k=4$, $t=3$ and $\eta=\gamma^{4}$. It's easy to get $k\neq 2t$ and $p\nmid k+1$. Thus $\alpha_{i}=\gamma^{\frac{q-1}{k}i}=\gamma^{20i}$ and $$\boldsymbol{\alpha}=\left(\gamma^{20+2},\gamma^{40+2},\gamma^{60+2},\gamma^{80+2}\right).$$ The corresponding GRL code $\mathcal{C}_{6}$ has the following generator matrix
		$$\begin{pmatrix}
			1&1&1&1&1\\
			\gamma^{22}&\gamma^{42}&\gamma^{62}&\gamma^{82}&0\\ 
			\left(\gamma^{22}\right)^2&\left(\gamma^{42}\right)^2&\left(\gamma^{62}\right)^2&\left(\gamma^{82}\right)^2&0\\
			\left(\gamma^{22}\right)^3&\left(\gamma^{42}\right)^3&\left(\gamma^{62}\right)^3&\left(\gamma^{82}\right)^3&\gamma^{4}
		\end{pmatrix}.$$
		Based on the Magma program, $\mathcal{C}_{6}$ is a LCD MDS code with the parameters $[5,4,2]_{3^4}$, which is consistent with Theorem \ref{EHullMGRS1}.
	\end{example}
	
	\subsection{An example for Theorem \ref{MGRSHhull}}
	\begin{example}\label{MGRSHhullexample}
		Let $q=9$, $\mathbb{F}_{q^2}^{*}=\langle\omega\rangle$ with $\omega^4=\omega^3+1$, $k=5$, $m=4$, $t=1$, $\eta=v_{n+1}=\omega$. By directly calculating, we know that the subgroup $U_{q+1}$ of $\mathbb{F}_{q^{2}}^{*}$ with order $q+1$ is
		$$U_{q+1}=\left\{\omega^0, \omega^8, \omega^{16}, \omega^{24}, \omega^{32}, \omega^{40}, \omega^{48}, \omega^{56}, \omega^{64}, \omega^{72}\right\}.$$
		And it is easy to get $\mathbb{F}_{q}^{*}=\langle \omega^{10}\rangle$. For convenience, we set $\gamma=\omega^{10}$, and so
		$$J_{q-1-m}=\left\{\gamma,\gamma^2,\gamma^3,\gamma^4\right\}\subseteq\mathbb{F}_{q}^{*},$$ thus
		$$I_{m}=\left\{c_{1},\ldots,c_{m}\right\}=\mathbb{F}_{q}^{*}\backslash J_{q-1-m}=\left\{\gamma^5,\gamma^6,\gamma^7,\gamma^8\right\}.$$ It's easy to know that 
		$$\sigma_{1}\left(J_{q-1-m}\right)=\gamma+\gamma^2+\gamma^3+\gamma^4=\gamma^2,$$
		$$\sigma_{2}\left(J_{q-1-m}\right)=\gamma\cdot\gamma^2+\gamma\cdot\gamma^3+\gamma\cdot\gamma^4+\gamma^2\cdot\gamma^3+\gamma^2\cdot\gamma^4+\gamma^3\cdot\gamma^4=1,$$
		$$\sigma_{3}\left(J_{q-1-m}\right)=\gamma\cdot\gamma^2\cdot\gamma^3+\gamma\cdot\gamma^2\cdot\gamma^4+\gamma\cdot\gamma^3\cdot\gamma^4+\gamma^2\cdot\gamma^3\cdot\gamma^4=\gamma^{7},$$
		$$\sigma_{4}\left(J_{q-1-m}\right)=\gamma\cdot\gamma^2\cdot\gamma^3\cdot\gamma^4=\gamma^{10},$$
		$$S_{1}\left(I_{m}\right)=\gamma^5+\gamma^6+\gamma^7+\gamma^8=\gamma^6,$$
		$$S_{2}\left(I_{m}\right)=\sum\limits_{t_{1}+t_{2}+t_{3}+t_{4}=2,t_{i}\geq 0}\gamma^{5t_{1}+6t_{2}+7t_{3}+8t_{4}}=1,$$
		$$S_{3}\left(I_{m}\right)=\sum\limits_{t_{1}+t_{2}+t_{3}+t_{4}=3,t_{i}\geq 0}\gamma^{5t_{1}+6t_{2}+7t_{3}+8t_{4}}=\gamma^{3},$$
		$$S_{4}\left(I_{m}\right)=\sum\limits_{t_{1}+t_{2}+t_{3}+t_{4}=4,t_{i}\geq 0}\gamma^{5t_{1}+6t_{2}+7t_{3}+8t_{4}}=\gamma^{2},$$
		Now by taking $I_{m}=\left\{c_{1},c_{2},\ldots,c_{m}\right\}=\left\{\gamma^5,\gamma^6,\gamma^7,\gamma^8 \right\}$ which satisfies
		$$S_{i}\left(I_{m}\right)\neq 0(i=1,\ldots,k-\delta).$$
		Furthermore, by taking
		$\beta_{j}\in\mathbb{F}_{q^2}^{*}(j=1,\ldots,m)$ with $c_{j}=\beta_{j}^{q+1}\in\mathbb{F}_{q}^{*}$, we have
		$$\left(\beta_{1},\beta_{2},\beta_{3},\beta_{4}\right)=\left(\omega^{5},\omega^{6},\omega^{7},\omega^{8}\right).$$
		And so, we can set $\boldsymbol{\alpha}=\left(\alpha_{1},\ldots,\alpha_{m(q+1)}\right)$ with
		\begin{equation}\label{alphai}
			\left\{\alpha_{1},\ldots,\alpha_{m(q+1)}\right\}=\bigcup\limits_{j=1}^{m}\beta_{j}\cdot U_{q+1}=\omega^{5}\cdot U_{q+1}\cup\omega^{6}\cdot U_{q+1}\cup\omega^{7}\cdot U_{q+1}\cup\omega^{8}\cdot U_{q+1}.
		\end{equation} 
		Now, by directly calculating $u_{j}=\prod\limits_{r=1, j \neq r}^{m}\left(c_{j}-c_{r}\right)^{-1}(j=1,2,3,4)$, we can get
		$$\left(u_{1},u_{2},u_{3},u_{4}\right)=\left(\gamma^{0},\gamma^{0},\gamma^{3},\gamma^{1}\right).$$ 
		Furthermore, by directly calculating $v_{(i-1)(q+1)+s}^{1+q}=u_{j}c_{j}^{m-\delta }\in\mathbb{F}_{q}^{*}(s=1,2,\ldots,q+1)$ with $\delta\in\left\{1,2,3\right\}$, respectively, we have the following three statements.
		
		$(1)$ If $\delta=1$, then
		\begin{equation}\label{vi1}
			\begin{aligned}
				&\left\{v_{1},\ldots,v_{q+1},\ldots,v_{(m-1)(q+1)+1},\ldots,v_{m(q+1)}\right\}\\
				=&\left\{\omega^{15},\ldots,\omega^{15},\omega^{18},\ldots,\omega^{18},\omega^{24},\ldots,\omega^{24},\omega^{25},\ldots,\omega^{25}\right\};
			\end{aligned}
		\end{equation}
		
		$(2)$ If $\delta=2$, then
		\begin{equation}\label{vi2}
			\begin{aligned}
				&\left\{v_{1},\ldots,v_{q+1},\ldots,v_{(m-1)(q+1)+1},\ldots,v_{m(q+1)}\right\}\\
				=&\left\{\omega^{10},\ldots,\omega^{10},\omega^{12},\ldots,\omega^{12},\omega^{17},\ldots,\omega^{17},\omega^{17},\ldots,\omega^{17}\right\};
			\end{aligned}
		\end{equation}  
		
		$(3)$ If $\delta=3$, then
		\begin{equation}\label{vi3}
			\begin{aligned}
				&\left\{v_{1},\ldots,v_{q+1},\ldots,v_{(m-1)(q+1)+1},\ldots,v_{m(q+1)}\right\}\\
				=&\left\{\omega^{5},\ldots,\omega^{5},\omega^{6},\ldots,\omega^{6},\omega^{10},\ldots,\omega^{10},\omega^{9},\ldots,\omega^{9}\right\}.
			\end{aligned}
		\end{equation}
		By substituting \eqref{alphai} and \eqref{vi1}, or \eqref{alphai} and \eqref{vi2}, or \eqref{alphai} and \eqref{vi3} into \eqref{Gvkt}, respectively, we can get the corresponding generator matrix(Omitted due to space constraints.). Basing on the Magma program, we have 
		$$	\boldsymbol{G}_{\boldsymbol{v},k,t}\boldsymbol{G}_{\boldsymbol{v},k,t}^{T_{h}}=\begin{pmatrix}
			\omega^{20}&\omega^{19}&0&0&0\\
			\omega^{11}&0&0&0&0\\
			0&0&1&0&0\\
			0&0&0&\omega^{30}&0\\
			0&0&0&0&\omega^{20}
		\end{pmatrix}$$
		or 
		$$	\boldsymbol{G}_{\boldsymbol{v},k,t}\boldsymbol{G}_{\boldsymbol{v},k,t}^{T_{h}}=\begin{pmatrix}
			\omega^{10}&\omega^{19}&0&0&0\\
			\omega^{11}&\omega^{70}&0&0&0\\
			0&0&\omega^{60}&0&0\\
			0&0&0&1&0\\
			0&0&0&0&\omega^{30}
		\end{pmatrix},$$
		or
		$$	\boldsymbol{G}_{\boldsymbol{v},k,t}\boldsymbol{G}_{\boldsymbol{v},k,t}^{T_{h}}=\begin{pmatrix}
			\omega^{10}&\omega^{19}&0&0&0\\
			\omega^{11}&\omega^{20}&0&0&0\\
			0&0&1&0&0\\
			0&0&0&\omega^{60}&0\\
			0&0&0&0&1
		\end{pmatrix},$$
		which are consistent with {\bf Case 1}, or {\bf Case 2.2}, or {\bf Case 2.1} of Theorem \ref{MGRSHhull}, respectively. 
		And the corresponding two classes of MGRS code are LCD NMDS with the parameters $[41,5,36]_{3^4}$, which are consistent with Theorem \ref{MGRSHhull}.
	\end{example}
	
	\section{Conclusions}\label{sec9} 
	In this paper, we focus on the MGRS code and its extended code. The main contributions are summarized as follows.
	\begin{itemize}
		\item Prove two classes of MGRS codes are either MDS or NMDS, present the necessary and sufficient conditions for these codes to be NMDS, and completely determine the weight distributions for a special class of NMDS MGRS codes.
		\item Prove two classes of EMGRS codes are either MDS or NMDS, present the necessary and sufficient conditions for these codes to be NMDS, and completely determine the weight distributions for a special class of NMDS EMGRS codes.
		\item Construct four classes of MGRS codes, which are either LCD or one-dimensional Euclidean hull codes.
		\item Construct a class of MGRS codes with flexible Hermitian hull dimensions and lengthss. 
		\item Prove the linearly
		inequivalence of NMDS MGRS codes and the linear code $\mathcal{C}(P, G, E)$ of elliptic-curve type from an
		elliptic curve $E$, where $P=\left\{P_{1},\ldots,P_{n}\right\},$ $G = kP_{0}.$
	\end{itemize}

\end{document}